\documentclass[twoside]{article}

\usepackage[accepted]{aistats2026} 

\usepackage[utf8]{inputenc} 
\usepackage[T1]{fontenc}    
\usepackage[round,authoryear]{natbib}
\usepackage{hyperref}       
\usepackage{url}            
\usepackage{booktabs}       
\usepackage{amsfonts}       
\usepackage{nicefrac}       
\usepackage{microtype}      
\usepackage{xcolor}         
\usepackage{amsmath} 
 \usepackage{graphicx}
 \usepackage{subfig}
 \usepackage{float} 
\usepackage{amsmath}
\usepackage{amssymb}
\usepackage{mathtools}
\usepackage{wrapfig}
\usepackage{amsthm}
\usepackage{subcaption}
\usepackage{multirow} 
\usepackage[table]{xcolor}  
\usepackage{colortbl}       
\usepackage{tabularx}
\usepackage{booktabs}
\usepackage{siunitx}
\usepackage{multirow}
\usepackage[table]{xcolor}
\usepackage{ragged2e}
\usepackage{caption}
\theoremstyle{plain}
\newtheorem{theorem}{Theorem}[section]

\theoremstyle{definition}

\theoremstyle{remark}

\usepackage{algorithm}
\usepackage{algorithmic}
\usepackage{pifont}
\definecolor{cvprblue}{rgb}{0.21,0.49,0.74}
\definecolor{greyL}{RGB}{230,248,255}

\begin{document}

\runningauthor{Dou et al.}
%

%

\twocolumn[

\aistatstitle{AMRM-Pure: Semantic-Preserving Adversarial Purification}

\aistatsauthor{
\begin{tabular}{c}
Zhihao Dou$^{1,2,*}$ \quad
Zhiqiang Gao$^{1,*,\dagger}$ \quad
Dongfei Cui$^{3}$ \quad
Weida Wang$^{4}$ \quad
Qinjian Zhao$^{1}$ \\
Dinggen Zhang$^{1}$ \quad
Jun Yan$^{5,\dagger}$ \quad
Zeke Xie$^{6}$ \quad
Shufei Zhang$^{7,\dagger}$
\end{tabular}
}

\aistatsaddress{
$^{1}$ Wenzhou-Kean University; 
$^{2}$ Case Western Reserve University; 
$^{3}$ Northeast Electric Power University;\\
$^{4}$ Fudan University; 
$^{5}$ Shanghai Ocean University; 
$^{6}$ HKUST (Guangzhou); 
$^{7}$ Shanghai AI Lab
}
\vspace{-2em}
\begin{center}
{\footnotesize
$^*$Equal contribution.\quad 
$^\dagger$Corresponding authors \quad \\
}
\end{center}

]


\begin{abstract}
  Adversarial purification is a defense technique that employs generative models to remove adversarial perturbations. Current methods often rely on powerful generators, typically diffusion models, and focus on reducing the gap between adversarial and clean samples in the feature space, while overlooking semantic correlation within a single sample. 
  To address this issue, we explore adversarial purification from the perspective of preserving semantic relationships among image patches.
  We employ an \textbf{A}ttentive \textbf{M}ask \textbf{R}econstruction \textbf{M}odel (\textbf{AMRM}), which shows superior performance. Our theoretical and experimental analysis reveals that AMRM is highly sensitive to adversarial noise, as such noise significantly distorts patch relationships. Based on this observation, we propose AMRM-Pure, a purification framework that denoises adversarial inputs by preserving patch-level semantics, and formulate this process as a tractable optimization problem with respect to the input. To further enhance robustness, we finetune AMRM-Pure with classification loss to strengthen semantic consistency. We apply our insight to two AMRM architectures, including Mask Autoencoder (MAE) and MaskDiT. Extensive experiments confirm the effectiveness of our method, establishing new state-of-the-art performance across multiple benchmarks.
\end{abstract}

\section{INTRODUCTION}

Deep Neural Networks (DNNs) are vulnerable to adversarial examples~\citep{12,13,14,15}, which are imperceptible to humans. However, these inputs with the malicious perturbations can cause DNNs to make erroneous predictions. Adversarial training~\citep{madry2018towards,zhang2019theoretically} is the state-of-the-art method for defending against adversarial attacks. However, the trade-off between generalization and robustness remains a concern~\citep{zhang2019theoretically}, especially against unseen adversarial examples. Furthermore, adversarial training incurs significantly higher computational costs compared to standard training.
\begin{figure*}[ht]

  \begin{minipage}[b]{0.15\linewidth}
    \centering
    \subfloat[\label{9o}]{\includegraphics[width=\linewidth]{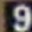}}
  \end{minipage}
   \hfill
   \begin{minipage}[b]{0.15\linewidth}
    \centering
    \subfloat[\label{9ma_a}]{\includegraphics[width=\linewidth]{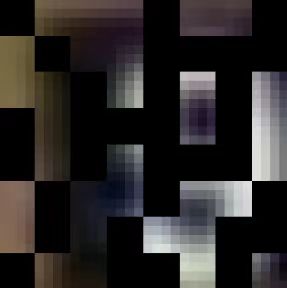}}
  \end{minipage}
   \hfill
  \begin{minipage}[b]{0.15\linewidth}
    \centering
    \subfloat[\label{9c}]{\includegraphics[width=\linewidth]{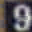}}
  \end{minipage}
   \hfill
  \begin{minipage}[b]{0.15\linewidth}
    \centering
    \subfloat[\label{adv9}]{\includegraphics[width=\linewidth]{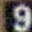}}
  \end{minipage}
   \hfill
  \begin{minipage}[b]{0.15\linewidth}
    \centering
    \subfloat[ \label{9advrec}]{\includegraphics[width=\linewidth]{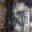}}
  \end{minipage}
  \hfill
  \begin{minipage}[b]{0.15\linewidth}
    \centering
    \subfloat[ \label{deadvb}]{\includegraphics[width=\linewidth]{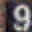}}
  \end{minipage}
  
  \caption{
  (a) Original image,
  (b) Masked image,
  (c) Clean image reconstruction from MAE,
  (d) Adversarial example under AutoAttack, 
  (e) Reconstruction of adversarial example under AutoAttack from MAE, 
  (f) Reconstruction of the denoised image under AutoAttack from MAE (denoised by our AMRM-Pure$_{\text{MAE}}$).}
  \label{effect}
  \vspace{-.15in}
\end{figure*}

Alternatively, another notable defense strategy is adversarial purification, which attracts widespread attention. Adversarial purification can be broadly classified into two categories, including purification with generative models \citep{yoon2021adversarial,nie2022diffusion,lin2024adversarial,bai2024diffusion,zhang2024classifier} and adaptation-based purification \citep{shi2021online}. Generative model-based approaches are the most widely used methods in adversarial purification, typically harnessing the powerful capabilities of generative models (e.g., diffusion) to transform the distribution of adversarial examples to that of clean samples \citep{nie2022diffusion}. 
Future efforts will aim to further enhance the denoising capabilities of the purification model through various approaches. These include leveraging contrastive guidance to steer diffusion models \citep{bai2024diffusion}, integrating classifier confidence guidance into the denoising process \citep{zhang2024classifier}, and fine-tuning the purification model with adversarial loss for robust optimization \citep{lin2024adversarial}.

The aforementioned methods primarily focus on aligning adversarial examples with the semantic distribution of clean samples, but neglect the semantic relationships among different patches within a sample. To fill this gap, we use an attentive mask reconstruction model (AMRM) to investigate how adversarial perturbations distort the semantic relationships among image patches. 
AMRM divides an image into multiple patches, masks a subset of them, and reconstructs the masked patches by using a self-attention \citep{vaswani2017attention} mechanism to explicitly capture dependencies among visible patches, such as the Mask Autoencoder (MAE) \citep{he2022masked} or MaskDiT \citep{zheng2024fast}.
%
To design a robust purification method, we first identify an intriguing phenomenon of the simple MAE so that an easier analysis.
%
Specifically, in the case of adversarial examples subjected to tiny, visually imperceptible perturbations, the reconstruction performance of MAE is severely compromised, dealing a devastating blow. As a typical example shown in Figure~\ref{9o} and\ref{adv9}, although the clean example and adversarial example appear very similar, MAE's reconstruction outputs in Figures~\ref{9c} and \ref{9advrec} exhibit significant differences. The reconstruction of perturbed data, as illustrated in Figure \ref{9advrec}, still displays poor quality. These findings suggest that the reconstruction capability of MAE is highly sensitive to adversarial perturbations, although these perturbations are visually imperceptible. Motivated by this observation, we consider preserving the semantic relationships among image patches as a novel mechanism for adversarial purification, a direction that has not been fully explored in existing works.

Based on such a research motivation, our proposed study aims to fill the gap. In this paper, through a series of analyses, we conjecture that this phenomenon results from adversarial perturbations that could easily distort semantic relations within patches, i.e., causing variation in the attention matrix (AMV), which leads to degraded image reconstruction quality of MAE. Through rigorous theoretical derivations and empirical experiments, we provide compelling evidence of the sensitivity of MAE to this AMV. 
Concretely, as the patch attention matrix essentially reflects how different semantic patches may be related to the masked patch, altering the attention matrix means a semantic change when the masked patch is reconstructed. Figures \ref{o_1} and \ref{adv_1} show that the reconstruction of the target patch in the red square depends on the similar patches in the clean image, while for adversarial images, high importance is assigned to distant and dissimilar patches. This suggests that the adversarial perturbations alter the semantic relations among patches.
Moreover, our findings reveal that the reconstruction loss of adversarial examples is lower-bounded by the sum of the loss for clean examples and the AMV. 
Meanwhile, we reveal that the sensitivity of AMV is transferable to a diffusion-based AMRM, e.g., MaskDiT, as shown in empirical analysis in Figure \ref{fig:combined_purification}.
%
Drawing inspiration from this finding, we propose a novel AMRM-Pure method which purifies adversarial perturbations by minimizing AMV, ultimately resulting in a inter-patch semantic preserving framework.
%

AMRM-Pure leverages the inherent sensitivity of AMV to adversarial noise, thereby achieving enhanced robustness. To realize AMRM-Pure, we introduce two variants: AMRM-Pure$_{\text{MAE}}$ based on MAE \citep{he2022masked} and AMRM-Pure$_{\text{MaskDiT}}$ based on MaskDiT \citep{zheng2024fast}. 
Furthermore, based on the insight of the previous work \citep{lin2024adversarial,zhang2024classifier},
we propose a Robust AMRM-Pure based on MAE (RAMRM-Pure$_{\text{MaskDiT}}$) and Robust MRM-Pure based on MaskDiT (RAMRM-Pure$_{\text{MaskDiT}}$) that leverages classification loss to fine-tune the purification model, significantly improving its inter-patch semantic preservation capability. We have extensively evaluated our method by comparing the important adversarial training and adversarial purification methods on various challenging adaptive attack benchmarks. Our method achieves state-of-the-art (SOTA) performance on four datasets, e.g., CIFAR-10 \citep{krizhevsky2009learning}, CIFAR-100 \citep{krizhevsky2009learning}, SVHN \citep{netzer2011reading}, and ImageNet \citep{deng2009imagenet}.

%


In summary, our main contributions are as follows:

1) We investigate the susceptibility of attention-based MRM to noise interference from both theoretical and empirical perspectives and disclose that the noise induces the deviation of semantic relations among patches, resulting in a degradation of the quality of the reconstruction. On the basis of our findings, we devise a novel and efficient purification technique, called AMRM-Pure, which is theoretically proven by rigorous analysis.

2) By successfully applying our approach to MAE and MaskDiT, we introduce AMRM-Pure$_{\text{MAE}}$ and AMRM-Pure$_{\text{MaskDiT}}$. Meanwhile, we further propose RAMRM-Pure$_{\text{MAE}}$ and RAMRM-Pure$_{\text{MaskDiT}}$, which incorporates classification loss to fine-tune the purification model, significantly improving standard and robust accuracy.

3) Extensive experiments have been conducted to validate the effectiveness of our MRM-Pure on various benchmarks, showing that our approach consistently achieves favorable outcomes after denoising processes.

\section{PRELIMINARIES AND RELATED WORK}
\label{sec:formatting}


\subsection{Adversarial Training} 
Adversarial training (AT) is a technique that enhances the robustness of a neural network by augmenting training samples with additional adversarial examples~\citep{goodfellow2014explaining,kurakin2016adversarial,tramer2017ensemble,zhang2019theoretically,dou2024improving}. 
%
%
Since AT typically involves a high computational cost, some studies~\citep{wong2020fast,liu2021using,vivek2020single} have focused on exploring ways to accelerate the training process by a one-step training strategy. In addition, the diffusion model has been employed for extensive data augmentation for adversarial training in many proposals~\citep{wang2023better,gowal2021improving,sehwag2021robust}, which enlarged the original dataset and enhanced the robust generalization.

\subsection{Adversarial Purification}
Generative models have shown great promise in purifying adversarial examples, drawing significant attention in robustness research. The early milestone study \citep{samangouei2018defense} introduced Defense-GAN, using GANs for purification. Song et al.~\citep{song2017pixeldefend} proposed the PixelDefense method, which employs the autoregressive models to mitigate the perturbations. Score-based generative models have also been applied for defense \citep{yoon2021adversarial}. Leveraging diffusion models, {DiffPure} \citep{nie2022diffusion} uses Stochastic Differential Equation (SDE) diffusion \citep{song2020score} for the denoising procedure, achieving robustness. Recent works \citep{li2024adbm,liu2025towards} further improve robustness by fine-tuning diffusion models. Lin et al.~\citep{lin2024adversarial} proposed a hybrid approach combining adversarial training with purification. It is significant to reconstruct the data without semantic information changes. 
IDC \citep{mei2025efficient} redesigns diffusion models from generating high-quality images to producing distinguishable label images, proposing an efficient image-to-image diffusion classifier that significantly reduces computational cost while improving adversarial robustness.
Thus, Bai et al.~\citep{bai2024diffusion} introduced contrastive guidance in diffusion models to enhance purification while preserving semantics. The adversarial purification method can be combined with other machine learning paradigms. For example, the framework of Self-supervised Online Adversarial Purification (SOAP) \citep{shi2021online} achieves notable results by integrating self-supervised tasks during training, further boosting robustness.

\begin{figure}[t]
    \centering
    \includegraphics[scale=0.24]{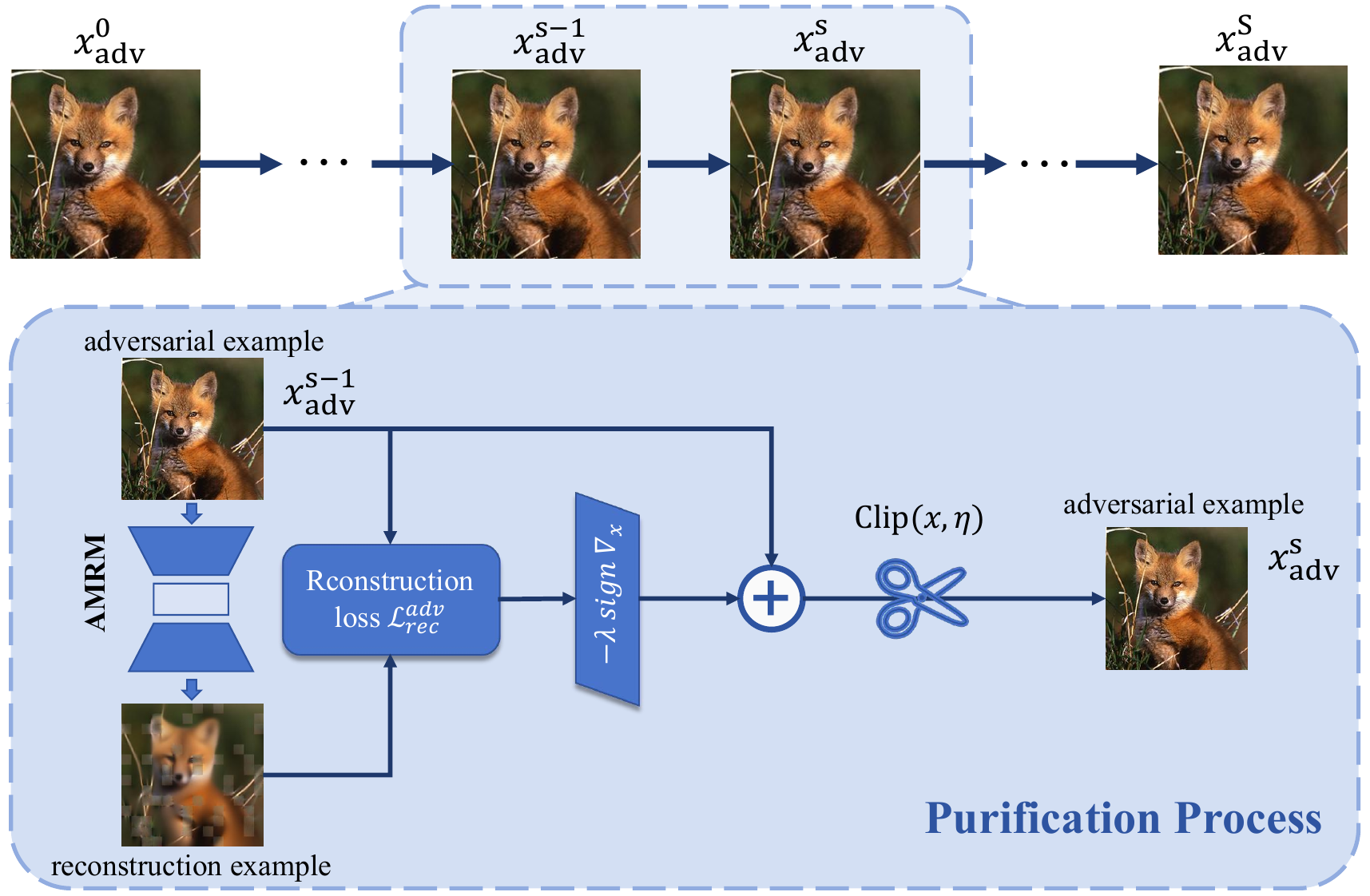}
    \caption{Overview of the proposed AMRM-Pure.}
    \label{fig:pipeline}
    \vspace{-0.5cm}
\end{figure}

There are several purification based on MAE \citep{zhou2023eliminating,wu2022denoising,you2023beyond}. DIR \citep{zhou2023eliminating} introduces a joint training framework of classifier and MAE under adversarial training, where the MAE restores robust features from unmasked patches to mitigate adversarial noise. Following the denoising autoencoder paradigm, DMAE \citep{wu2022denoising} and NIM-MAE \citep{you2023beyond} incorporate Gaussian noise into masked image modeling, which is eliminated through the encoding–decoding process. Specifically, DMAE focuses on achieving robust pre-training against Gaussian noise to enhance generalization and robustness, though it shows limited effectiveness against adversarial perturbations. NIM-MAE leverages pre-trained models to remove adversarial noise, yet its robustness performance still leaves room for improvement.
Unlike previous methods, ARM-Pure leverages the sensitivity of semantic patch relationships to adversarial perturbations and employs optimization-based denoising to reduce AMV, effectively minimizing semantic variations in adversarial examples. This novel perspective has not been studied in previous research.


\begin{figure*}[t]
    \begin{minipage}[b]{0.23\linewidth}
    \centering
    \subfloat[\label{m_1}]{\includegraphics[width=\linewidth]{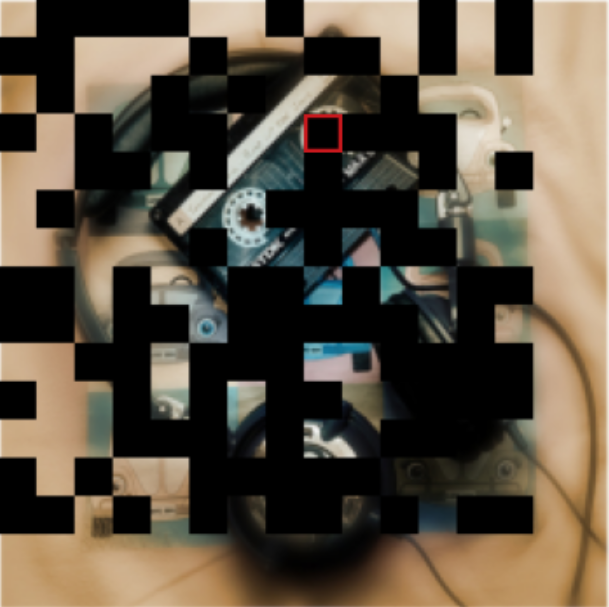}}
  \end{minipage}
   \hfill
  \begin{minipage}[b]{0.23\linewidth}
    \centering
    \subfloat[\label{o_1}]{\includegraphics[width=\linewidth]{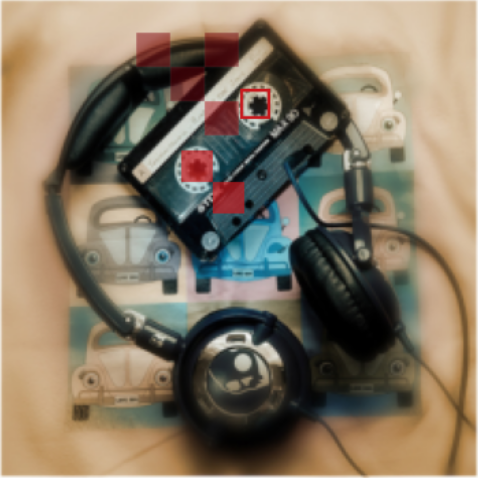}}
  \end{minipage}
   \hfill
  \begin{minipage}[b]{0.23\linewidth}
    \centering
    \subfloat[\label{adv_1}]{\includegraphics[width=\linewidth]{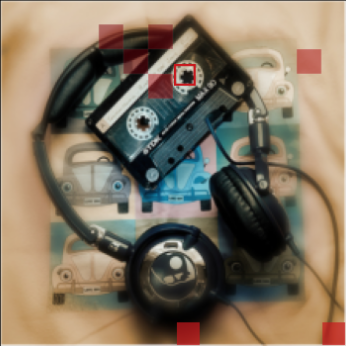}}
  \end{minipage}
  \hfill
  \begin{minipage}[b]{0.23\linewidth}
    \centering
    \subfloat[\label{pur_1}]{\includegraphics[width=\linewidth]{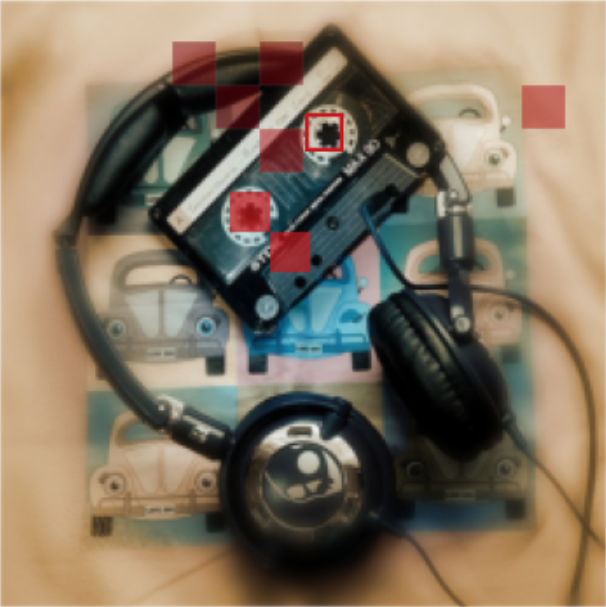}}
  \end{minipage}
  \hfill
   \begin{minipage}[b]{0.038\linewidth}
    \centering
    \vspace*{0.01cm}
    {\includegraphics[width=\linewidth]{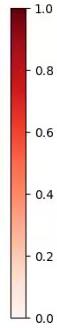}}
  \end{minipage}
  \hfill
  \caption{
The first column, Figure (a) represents the Mask Matrix.
The second column, Figure (b) illustrates the Attention Weights for clean samples.
The third column, Figure (c) depicts the Attention Weights for adversarial examples.
The fourth column, Figure (d) showcases the Attention Weights for denoised samples (by our AMRM-Pure$_{\text{MAE}}$). 
  Patches with a deeper red color mean the elements with more attention. The data is sampled from the ImageNet dataset~\citep{deng2009imagenet}.}
  \label{amv}
    \vspace{-.15in}
\end{figure*}

\subsection{Preliminary of Masked Autoencoder (MAE)} \label{MAE_pre}
MAE ~\citep{he2022masked} is briefly introduced within the context of adversarial robustness in this subsection. A clean input sample \( x \), drawn from the dataset \( X \), is partitioned into \( n \) patch vectors of dimension \( d \), forming \( \bar{x} \in \mathbb{R}^{n \times d} \). The matrix \( \bar{x} \) can be randomly divided into \( m = (1-\rho)n \) masked patch vectors and \( (n-m) \) visible patch vectors, where \( \rho \) is the mask ratio.
MAE uses an encoder-decoder architecture. The encoder, \( f(\cdot) \), produces \(\mathbf{V}^{enc} \in \mathbb{R}^{m \times d_e}\), where \(\mathbf{V}^{enc} = f(x_1)\), and \(x_1\) is the visible portion of input \(x\). Here, \(d_e\) is the dimension of each patch feature in \(\mathbf{V}^{enc}\). The decoder, \(g(\cdot)\), maps \(\mathbf{V}^{enc}\) back to pixel space, producing \(\mathbf{V}^{dec} \in \mathbb{R}^{(n-m) \times d}\), i.e., \(g(\mathbf{V}^{enc}) = \mathbf{V}^{dec}\), which reconstructs masked patches \(x_2\). Reconstruction quality is measured using Mean Squared Error (MSE) loss as follows:
\begin{equation}
    \mathcal{L}_{\text{rec}}(x_1)=\frac{1}{N(n-m)}\sum_{i=1}^{N}||g(f(x_{1,i}))-x_{2,i}||^{2},
    \label{mae:loss}
\end{equation}
where $N$ represents the sample number in $X$.

The MAE structure consists of multiple self-attention layers, where attention captures semantic relationships between input patches. At the \(t\)-th layer, the input features are \(\mathbf{Z}^t \in \mathbb{R}^{n_t \times d_t}\), with \(n_t\) patches and \(d_t\)-dimensional patch features. Weight matrices \(\mathbf{W}_Q^t\), \(\mathbf{W}_K^t\), and \(\mathbf{W}_V^t\) generate the query \(\mathbf{Q}^t\), key \(\mathbf{K}^t\), and value \(\mathbf{V}^t\) matrices, all in \(\mathbb{R}^{n_t \times d_t}\).

The self-attention matrix \(\mathbf{A}^t\) is computed as:
\[
\mathbf{A}^t = \text{softmax}\left(\frac{\mathbf{Q}^t (\mathbf{K}^t)^T}{\sqrt{d_t}}\right),
\]
quantifying similarities between \(\mathbf{Q}^t\) and \(\mathbf{K}^t\). The \(j\)-th output patch feature \(e_j\) is a weighted sum of value vectors:
\[
e_j = \sum_{i=1}^n a_{ji}^t v_i^t, \quad a_{ji}^t = \frac{q_{ji}^t k_{ji}^t}{\sum_{o=1}^n q_{jo}^t},
\]
where \(a_{ji}^t\) indicates how much \(v_i^t\) contributes to \(e_j\).

\section{THEORETICAL ANALYSIS} \label{Sec:Theory}

In this section, we theoretically analyze how adversarial perturbations affect semantic relationships among patches.
We take MAE as a representative AMRM architecture due to its simplicity and typical design, and investigate the link between AMV variation and decoder reconstruction loss, both theoretically and empirically.
Given the structural similarity, we further extend this analysis to MaskDiT (Sec. \ref{maskdit}), and our results confirm that the theory also holds, demonstrating strong transferability.

\subsection{Adversarial Perturbation Induces Attention Matrix Variation in MAE}
\label{sec:senstively}

Given a clean sample \(x\) and its adversarial counterpart \(x_{adv}\), the attention matrices and input features at the \(t\)-th layer in MAE are \(\mathbf{A}^t\) and \(\mathbf{Z}^t\) for \(x\), and \(\mathbf{A}_{adv}^t\) and \(\mathbf{Z}_{adv}^t\) for \(x_{adv}\), respectively. To save the space, more defination can be founded in Appendix \ref{MAE_pre}. The attention matrix variation (AMV) at layer $t$ induced by adversarial perturbation is formally defined as $\mathbf{A}_{adv}^t - \mathbf{A}^t$, where \(\mathbf{W}_K^t, \mathbf{W}_Q^t \in \mathbb{R}^{n_t \times d_t}\) are the weight matrices at the \(t\)-th layer, \textcolor{black}{and \(N\) represents} the number of training samples. AMV indicates a shift in MAE’s focal points on the image, reflecting a change in the inter-patch semantic information being captured, as $\mathbf{A}^t_{adv}$ misaligns attention toward irrelevant regions and distorts the overall interpretative context (see Figure~\ref{amv}).
To quantify how such noise affects the attention matrix, we derive Theorem~3.1 to formally express the impact of perturbation $\delta_t$ on AMV, revealing the inherent AMV sensitivity of the MAE.

\begin{theorem}	
\label{thm1}
      Let $\delta_t = \textbf{Z}_{adv}^t - \textbf{Z}^t$ denotes the latent feature shift caused by the adversarial perturbation at layer $t$ in MAE. With a set $\{\omega_i\}_{i=0}^{k}$ and kernel coefficient $\omega_i \in \mathcal{N}\left(0, \mathbf{I}_{d}\right)$, it holds that:
      {\scriptsize 
      \begin{align*}
    ||\mathbf{A}_{adv}^{t}-\mathbf{A}^{t}||_2 &\ge
    \gamma \left\| \left[\left( \mathbf{Y} - \mathbf{B} \mathbf{Q}^t \right)^\top \mathbf{W}^t_Q + \left( \mathbf{Y} - \mathbf{B} \mathbf{K}^t \right)^\top \mathbf{W}^t_K \right]\delta_t \right\|_2,
\end{align*}

\begin{align*}
\mathbf{B} &= \sum_{i=0}^{k} \exp\!\left( \omega_i^\top (\mathbf{Q}^t + \mathbf{K}^t) \right), &
\mathbf{Y} &= \sum_{i=0}^{k} \exp\!\left( \omega_i^\top (\mathbf{Q}^t + \mathbf{K}^t) \right)\omega_i, \\
\gamma &= \frac{\exp\!\left( -\tfrac{\|\mathbf{Q}^t\|^2 + \|\mathbf{K}^t\|^2}{2} \right)}{m}.
\end{align*}

      } 
\end{theorem}
\begin{proof}
The proof can be seen in Appendix \ref{proof_thm1}.
\end{proof}

Theorem~\ref{thm1} suggests that even minor shifts in the latent features ($\delta_t$) may have the ability to cause disproportionately large changes in AMV, especially due to the high dimensionality $d_t$ of internal projection matrices. Notably, in MAE/AMRM, where $d_t \gg \text{input dimension}$, the sensitivity is further amplified. This analysis reveals the intrinsic vulnerability of MAE’s attention mechanism under adversarial conditions, offering a theoretical foundation for the empirical trends shown in Figure \ref{trens_forard} (a-c) of Section \ref{sec:emp}.

\subsection{Impact of Decoder Attention Shifts on Adversarial Reconstruction in MAE}

To deepen the understanding of how attention pattern distortions affect output quality in MAE, we present a theoretical lower bound on the reconstruction loss under adversarial conditions. This analysis extends the discussion of AMV sensitivity in Section 3.1 and reveals how attention shifts in the decoder layer affect reconstruction loss.
First, let $\mathcal{L}_{\text{rec}}^{adv}$ denote the average reconstruction loss for adversarial examples, corresponding to the reconstruction loss $\mathcal{L}_{\text{rec}}$ for clean samples as Eq.~(\ref{mae:loss}) in Appendix \ref{MAE_pre}.

\begin{theorem}	
\label{thm2}
Let \(\mathbf{A}^{\mathrm{dec}}_{i,t}\) denote the attention matrix at the \(t\)-th layer of the MAE decoder for the \(i\)-th sample in the dataset, and let \(\mathbf{A}^{\mathrm{dec}}_{\mathrm{adv},i,t}\) denote the corresponding attention matrix for the adversarial examples. With ratio constants \(C_A\) and \(H\), it holds that:
{\footnotesize 
\begin{align*}
    \mathcal{L}^{\mathrm{adv}}_{\mathrm{\text{rec}}} \geq \frac{1}{2} \mathcal{L}_{\mathrm{\text{rec}}} + \frac{1}{2NT} \sum_{t=1}^{T} \sum_{i=1}^{N} \left[ H C_{A} \left\| \mathbf{A}^{\mathrm{dec}}_{\mathrm{adv},i,t} - \mathbf{A}^{\mathrm{dec}}_{i,t} \right\|^2 - c_{rec} \right].
\end{align*}
} 
$c_{rec}$ is the reconstruction bias, which symbolizes the disparity between the output of MAE and the original, unmasked image. The definition of $c_{rec}$ can be found in Appendix \ref{exp:assumption}. 
\end{theorem}

\begin{proof}
The proof can be seen in the Appendix \ref{proof:thm2}.
\end{proof}

It shows that the lower bound of the reconstruction loss for adversarial data can be decomposed into three components: the average reconstruction loss for clean data $\mathcal{L}_{\text{rec}}$, the average attention matrix variation for the MAE decoder at each layer $\frac{1}{2NT}\sum_{t=1}^{T}\sum_{i=1}^{N}||(A^{dec}_{adv_{i,t}}-A^{dec}_{i,t})||^2$, and constant terms. 
Theorem~3.2 demonstrates that adversarial distortions in AMV of decoder lead to increases in reconstruction loss $\mathcal{L}_{\text{rec}}$, and Figure~\ref{trens_forard} provides further evidence of this phenomenon. Notably, the reconstruction loss is shown to be consistent with the degree of AMV, confirming a strong correlation between inter-patch semantic relationships and output degradation.
This theorem builds a theoretical foundation of robust MAE-based purification.

\subsection{Empirical Validation}
\label{sec:emp}


This section empirically validates our theoretical analysis.
We first assess the impact of MAE reconstruction on adversarial perturbations, then analyze the correlation between attention matrix variation and reconstruction loss.
Finally, we show that adversarial perturbations alter semantic relationships within MAE patches and confirm the sensitivity of AMV to such perturbations.
\begin{figure*}
    \centering
    \includegraphics[width=1.0\linewidth]{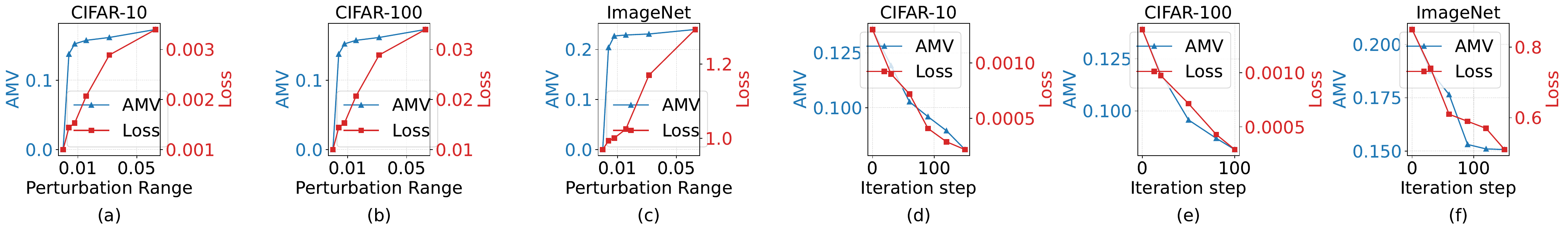} 
    \caption{Trends of MAE reconstruction loss and attention matrix variation under AutoAttack and during purification across multiple datasets. (a–c): Under AutoAttack on CIFAR-10, CIFAR-100, and ImageNet with different attack budget. (d–f): During the purification process with AMRM-Pure$_{\text{MAE}}$ on CIFAR-10, CIFAR-100, and ImageNet. }
    \label{trens_forard}
    \vspace{-0.15cm}
\end{figure*}

\textbf{Perturbation Leads to Degraded Reconstruction Quality.} To empirically investigate how perturbation influences reconstruction, an image is randomly selected from the SVHN dataset. We then compare the reconstruction results of the clean data and the adversarial example as shown in Figure~\ref{effect}. 
The adversarial example generated by the AutoAttack procedure~\citep{croce2020reliable} in Figure \ref{adv9} looks almost identical to its clean counterpart in Figure~\ref{9o} visually.  However, its reconstruction (Figure \ref{9advrec}) is significantly different from that of the clean sample (Figure \ref{9c}).
Likewise, its reconstruction result also shows significant differences with a reconstruction of its clean sample (Figure~\ref{9c}). 
These phenomena emphasize the substantial impact of adversarial perturbations on the overall outcome of the reconstruction.

\textbf{Visualization of Attention Matrix Variation.}
To check how the semantic relationship between different patches changes under perturbation, we show the degree of importance of visible patches for reconstructing the masked patch in Figure~\ref{amv}.
We select an image from ImageNet~\citep{deng2009imagenet} and randomly generate a mask matrix. A specific masked patch is considered as the target patch (marked red in a box) in Figure~\ref{m_1}.
%
%
{Then, the visible patches are fed into the MAE to perform the reconstruction. The degree of importance of visible patches for reconstructing the target patch, which is determined by the corresponding values of the attention matrix, is illustrated in Figure~\ref{o_1} (e.g., the last layer attention of the MAE decoder).}
Meanwhile, we also display the corresponding visualizations of the adversarial example and the denoised example for the same mask matrix and target patch position.

%
In generating attention weights for the target patch (red box), the approach involves using the patch itself as the query vector and the remaining patches as key vectors. Through self-attention, the attention weights are determined. Higher weights indicate a stronger semantic similarity between the target patch and the patch itself. 
Figures~\ref{o_1}, \ref{adv_1}, and \ref{pur_1} show the importance degree of visible patches in the clean sample, adversarial example, and denoised example, respectively. Patches that are closer to red are more significant.
As demonstrated in Figure~\ref{o_1}, when a hole in the tape is used as the target patch, the visible patch of another similar hole and the surrounding patches serve as the most important basis for reconstructing the target patch.
%
However, as illustrated in Figure~\ref{adv_1}, some distant and irrelevant visible patches with large color and shape differences are taken or focused to reconstruct the adversarial perturbed target patch. This indicates that the adversarial perturbation leads MAE to erroneous attention. More examples can be seen in Appendix \ref{visualization_more}.

\textbf{Analysis of AMV Sensitivity and Consistency with Reconstruction Loss.}
To validate Theorem \ref{thm2} and the effect of AMV on the reconstruction quality, we visualize the changes of reconstruction loss and AMV values with respect to the intensity of adversarial noise in Figure \ref{trens_forard}.
We evaluate reconstruction loss and average AMV $\frac{1}{N}\sum_{i=1}^N||\mathbf{A}^{dec}_{adv,i}-\mathbf{A}^{dec}_i||_2$ on 150 images from CIFAR-10, CIFAR-100, and ImageNet with a 0.5 mask ratio, using AutoAttack to generate adversarial examples.
As shown in Figure~\ref{trens_forard} (a–c), AMV is highly sensitive, rising sharply even under small perturbations (e.g., $\delta=0.01$), and then steadily increasing until reaching its softmax-bounded upper limit.
Reconstruction loss and AMV follow the same trend, aligning with our theoretical analysis in Theorem \ref{thm2} of their relationship and AMV’s sensitivity to perturbations. 


\subsection{A General AMRM Framework}
\label{maskdit}

The aforementioned analysis based on simple MAE reveals \textit{a key vulnerability in attentive masked image modeling: inter-patch semantic information (captured by AMV) is sensitive to noise.} 
To show the universality of this observation, we further illustrate this sensitivity of AMV on a more powerful generative model, such as MaskDiT~\citep{zheng2024fast}, which is a powerful diffusion-based AMRM that preserves its masked autoencoder architecture while integrating forward and reverse diffusion steps for high-quality image generation. 
Given their architectural alignment, the MaskDiT indeed agrees with the analysis presented in Theorem \ref{thm1} in Sec \ref{sec:senstively}.
As shown in Figure~\ref{fig:combined_purification_maskdit} of Appendix \ref{maskdit_val}, empirical comparisons show that MaskDiT exhibits consistent patterns with MAE under adversarial settings.
Specifically, the AMV of MaskDiT remains highly sensitive to noise, so that its reconstruction loss and AMV consistently maintain alignment. As such, MaskDiT continues to satisfy Theorem \ref{thm1} and \ref{thm2}.
Motivated by this phenomenon, we propose AMRM-Pure$_{\text{MaskDiT}}$, which leverages MaskDiT's generative capacity while retaining MAE's inter-patch semantic sensitivity, further improving purification performance by optimizing its reconstruction loss (Eq.~(\ref{reconst_update})).

\section{METHOD} \label{Sec:method}

\subsection{Adversarial Purification with AMRM}


As discussed in Section~\ref{Sec:Theory}, the attention mechanism of AMRM  is highly sensitive to adversarial perturbations, which distort inter-patch semantic relationships and degrade reconstruction quality. Based on this sensitivity, we propose a purification scheme, AMRM-Pure, which formulates denoising as an optimization problem that minimizes semantic variations.

We denote the clean data as $x$, and the adversarial example as $x_{adv}$. 
In the context of denoising, the objective is to induce a modification $\Delta$ on $x_{adv}$ such that the attention matrix of the denoised examples, $atten(x_{adv}+\Delta)$, closely aligns with the attention matrix of clean data samples $atten(x)$. 
Therefore, the learning objective of denoising can be formed as an Attention Matrix Variation Minimization problem, which is denoted as:
\begin{equation}
\begin{aligned}
    &\min_{\Delta} \mathcal{L}(\Delta) = ||\text{atten}(x_{adv}+\Delta) - \text{atten}(x)||_2,\\     
    &\hspace{0.2cm} \text{s.t.} \hspace{0.2cm} ||\Delta||_{\infty} \leq C_e,
\end{aligned} 
\label{eq:1}
\end{equation}
where $C_e$ is a small constant.

In addressing the AMV Minimization problem, the denoising process strives to mitigate adversarial perturbations on $x_{\text{adv}}$. Its goal is to achieve a consistent attention matrix between the denoised images and the clean images of MAE. This alignment ensures that the patch relationship of the denoised image closely approximates that of clean samples, thereby minimizing the impact of adversarial perturbations in the denoising outcome.


Since clean attention $atten(x)$ is unavailable during inference, directly minimizing AMV Eq.~(\ref{eq:1}) is intractable. 
Instead, inspired by Theorem~\ref{thm2}, we minimize the AMRM reconstruction loss as a tractable surrogate. Reconstruction loss not only enables efficient optimization without requiring clean attention (tractability), but also can reduce AMV caused by perturbations (owing to the consistent trend presented in Theorem~\ref{thm2}), thereby aligning with the inter-patch semantic structure of the clean input.
As such, our approach instead minimizes the AMRM reconstruction loss for adversarial purification. 
\begin{align}
    &min_{\Delta} \mathcal{L}_{\text{rec}}(x_{adv}+\Delta) \iff min_{\Delta}\mathcal{L}(\Delta), 
    \nonumber\\
 &\quad\qquad s.t.\quad ||\Delta||_{\infty} \leq C_e.
\label{reconst_update}
\end{align}
To address this problem, we employ the standard Projected Gradient Descent (PGD) method~\citep{madry2018towards}. 
In this approach, the modifications are iteratively added to adversarial examples, and the total number of iterations is denoted as $S$. At the $s$-th iteration, the denoising process is denoted as:
\begin{equation}
\begin{aligned}
x_{adv}^{s} &= Clip(x_{adv}^{s-1} - \lambda \cdot \Delta_{s},\eta),\\
\Delta_s &= \text{sign} (\nabla_{x}  L_{\text{rec}}(x_{adv}^{s-1})).    
\end{aligned} 
\label{results}
\end{equation}

\begin{table}[ht]
\centering
\setlength{\tabcolsep}{3pt} 
\renewcommand{\arraystretch}{0.9}
\tiny
\caption{Clean and robust accuracy (\%) on CIFAR-10 obtained by different purification methods. WideResNet is commonly abbreviated as WRN.}
\begin{tabular}{c|c|ccc}
\toprule
\multirow{2}{*}{Method}          & \multirow{2}{*}{Classifier} & \multirow{2}{*}{Std Acc} & \multicolumn{2}{c}{Robust Acc}   \\ \cline{4-5} 
                                 &                               &                          & $\ell_{\infty}$ & $\ell_{2}$     \\ \midrule
Shi et al.~\citep{shi2021online}             & WRN-28-10             & 91.89                    & 4.56            & 7.25           \\
Yoon et al.~\citep{yoon2021adversarial}       & WRN-70-16             & 87.93                    & 37.65           & 57.81          \\
Zhang et al.~\citep{zhang2023enhancing}       & WRN-70-16             & \textbf{93.16}           & 22.07           & 35.74          \\
Diffpure \citep{nie2022diffusion}            & WRN-70-16             & 92.50                    & 42.20           & 60.80          \\
COUP \citep{zhang2024classifier}             & WRN-28-10             & 90.33                    & 41.72           & 57.25          \\
ADBM \citep{li2024adbm}                      & WRN-70-16             & 91.90                    & 47.70           & 63.30          \\
ADDT$_{w/ \text{Diffpure}}$ \citep{liu2025towards} & WRN-28-10     & 89.94                    & 55.76           & -              \\ \midrule
\cellcolor{greyL}AMRM-Pure$_{\text{MAE}}$       & \cellcolor{greyL}WRN-28-10 & \cellcolor{greyL}88.57  & \cellcolor{greyL}40.53  & \cellcolor{greyL}53.50  \\
\cellcolor{greyL}AMRM-Pure$_{\text{MaskDiT}}$    & \cellcolor{greyL}WRN-28-10 & \cellcolor{greyL}92.03  & \cellcolor{greyL}50.57  & \cellcolor{greyL}64.53 \\
\cellcolor{greyL}RAMRM-Pure$_{\text{MAE}}$      & \cellcolor{greyL}WRN-28-10 & \cellcolor{greyL}90.09  & \cellcolor{greyL}45.15  & \cellcolor{greyL}60.72  \\
\cellcolor{greyL}RAMRM-Pure$_{\text{MaskDiT}}$  & \cellcolor{greyL}WRN-28-10 & \cellcolor{greyL}93.11 & \cellcolor{greyL}\textbf{62.13}  & \cellcolor{greyL}\textbf{73.57}  \\
\bottomrule
\end{tabular}
\label{cifar_10_full}
\end{table}

Here $\mathcal{L}_{\text{rec}}$ signifies the AMRM reconstruction loss defined in Eq.~(\ref{mae:loss}), with $\lambda$ representing the step size and $\eta$ as the clipping threshold. The overall modification $\Delta$ is composed of individual iteration modification $\Delta_s$ for $s \in [1,  S]$. The purpose of $\Delta$ is to guide the attention distribution of adversarial examples $atten(x_{adv})$ towards the clean sample distribution $atten(x)$. The denoising algorithm and pipeline of our AMRM-Pure can be seen in Algorithm \ref{alg1} (Supplement \ref{sec:algo}) and the whole process is in Figure \ref{fig:pipeline}. 
For AMRM-Pure, when the mechanism is applied to the MAE framework, we denote it as AMRM-Pure$_{\text{MAE}}$; when applied to the MaskDiT framework, we denote it as AMRM-Pure$_{\text{MaskDiT}}$.
To empirically validate the theory of AMRM-Pure, we also plot the trends of the MAE reconstruction loss and AMV with increasing denoising iterations in Figure~\ref{trens_forard} (d-f) and Fig \ref{fig:combined_purification_maskdit} (b)  for MaskDiT.  We randomly selected 100 adversarial examples from each dataset, perturbed using AutoAttack. For CIFAR10 and CIFAR100, the perturbation magnitude is set to $\ell_{\infty}=\frac{8}{255}$, while for ImageNet, it is set to $\ell_{\infty}=\frac{4}{255}$. As observed,  both the AMV and loss exhibit a similar downward trend as the number of purification iterations increases. This supports the validity of our theory.

Furthermore, we provide strict convergence analysis within the Appendix \ref{Convergence}.





\subsection{Robust Purification Model}

As the study of AToP \citep{lin2024adversarial} shows, further fine-tuning a purification model using classification loss can enhance its robustness against both seen and unseen attacks.
Following this insight, we propose a two-stage fine-tuning method to develop Robust AMRM-Pure$_{\text{MAE}}$ (RAMRM-Pure$_{\text{MAE}}$) and Robust AMRM-Pure$_{\text{MaskDiT}}$ (RAMRM-Pure$_{\text{MaskDiT}}$) variants to enhance the semantic relationship-preserving capabilities. For more details about our method, refer to Appendix \ref{robust:purification}.

Figure~\ref{fig:purification_iterations}  (see Appendix \ref{robust:purification}) shows a quantitative analysis of enhanced semantic relationships using AMV as an evaluation metric.
We randomly select 100 images from the CIFAR-10 dataset and examined the AMV of AMRM-Pure$_{\text{MAE}}$ and RAMRM-Pure$_{\text{MAE}}$ under the AutoAttack with an attack budget of $\frac{8}{255}$ across different purification iterations.
Specifically, the initial AMV (without purification) of RAMRM-Pure$_{\text{MAE}}$ is higher than that of AMRM-Pure$_{\text{MAE}}$. This suggests that Robust MAE is more sensitive to interpatch semantic information changes caused by adversarial attacks. However, with the progression of purification iterations, the AMV of RAMRM-Pure$_{\text{MAE}}$ decreases significantly, highlighting its superior capability in preserving semantic integrity compared to AMRM-Pure$_{\text{MAE}}$. 
Furthermore, as shown in Figure~\ref{fig:purification_iterations_md} (Appendix~\ref{robust:purification}), both AMRM-Pure$_{\text{MaskDiT}}$ and RAMRM-Pure$_{\text{MaskDiT}}$ exhibit similar trends.

\section{EXPERIMENT}
\label{sec:exp}

\subsection{Experimental Setting}


\noindent \textbf{Datasets and Classifier.}
In this section, we validate the robustness of our purification method, AMRM-Pure, on four benchmark datasets, including CIFAR-10 \citep{krizhevsky2009learning}, CIFAR-100 \citep{krizhevsky2009learning}, SVHN \citep{netzer2011reading}, and ImageNet \citep{deng2009imagenet}. We use WideResNet-28-10~\citep{zagoruyko2016wide} as the main classifier for CIFAR-10, CIFAR-100, and SVHN, and ResNet-101~\citep{he2016deep} as the main classifier for ImageNet.

\begin{table}[t]
\tiny
\centering
\caption{Clean and robust accuracy (\%) on CIFAR-100 obtained by different purification methods. The experiment are implemented on WideResNet-28-10.}
\label{cifar_100}
\begin{tabular}{c|ccc}
\midrule[1.2pt]
\multirow{2}{*}{Method} & \multirow{2}{*}{Std Acc} & \multicolumn{2}{c}{Robust Acc} \\ \cline{3-4} 
                        &                          & $\ell_{\infty}$   & $\ell_{2}$  \\ \midrule
Diffpure \citep{nie2022diffusion} & 45.23                    & 11.57           & 31.53         \\
COUP \citep{zhang2024classifier}  & 65.71                    & 15.22           & 34.28         \\
ADDT$_{w/ \text{DDPM}}$ \citep{liu2025towards}  & 66.02                    & 18.85           & 36.57         \\ \hline
\cellcolor{greyL}AMRM-Pure$_{\text{MAE}}$                         & \cellcolor{greyL}65.34                    & \cellcolor{greyL}14.28           & \cellcolor{greyL}29.29         \\ 
\cellcolor{greyL}AMRM-Pure$_{\text{MaskDiT}}$                     & \cellcolor{greyL}\textbf{70.03}           & \cellcolor{greyL}24.39           & \cellcolor{greyL}36.51         \\
\cellcolor{greyL}RAMRM-Pure$_{\text{MAE}}$                        & \cellcolor{greyL}66.28                    & \cellcolor{greyL}19.53           & \cellcolor{greyL}31.58         \\
\cellcolor{greyL}RAMRM-Pure$_{\text{MaskDiT}}$                     & \cellcolor{greyL}69.87                    & \cellcolor{greyL}\textbf{29.91}  & \cellcolor{greyL}\textbf{43.27} \\ \midrule[1.2pt]
\end{tabular}
\label{tab:cifar100}
\end{table}

\begin{table}[t]
\tiny
\centering
\caption{Clean and robust accuracy (\%) on SVHN obtained by different purification methods. The experiment are implemented on WideResNet-28-10.}
\label{svhn}
\begin{tabular}{c|ccc}
\midrule[1.2pt]
\multirow{2}{*}{Method} & \multirow{2}{*}{Std Acc} & \multicolumn{2}{c}{Robust Acc} \\ \cline{3-4} 
                        &                          & $\ell_{\infty}$   & $\ell_{2}$  \\ \midrule
Diffpure \citep{nie2022diffusion} & 93.90                    & 39.70           & 63.30            \\
COUP \citep{zhang2024classifier}  & 92.07                    & 41.62           & 63.97             \\
ADBM \citep{li2024adbm}           & 93.50                    & 47.90           & 65.70             \\ \hline
\cellcolor{greyL}AMRM-Pure$_{\text{MAE}}$                         & \cellcolor{greyL}94.54                    & \cellcolor{greyL}27.59           & \cellcolor{greyL}55.29           \\
\cellcolor{greyL}AMRM-Pure$_{\text{MaskDiT}}$                     & \cellcolor{greyL}94.91                    & \cellcolor{greyL}46.57           & \cellcolor{greyL}66.38          \\
\cellcolor{greyL}RAMRM-Pure$_{\text{MAE}}$                        & \cellcolor{greyL}94.47                    & \cellcolor{greyL}39.15           & \cellcolor{greyL}60.51            \\
\cellcolor{greyL}RAMRM-Pure$_{\text{MaskDiT}}$                     & \cellcolor{greyL}\textbf{95.39}                  & \cellcolor{greyL}\textbf{55.90}  & \cellcolor{greyL}\textbf{70.18}  \\ \midrule[1.2pt]
\end{tabular}
\label{tab:svhn}
\end{table}

\noindent \textbf{Adversarial Attacks.}
Several studies \citep{chen2023robust,li2024adbm,liu2025towards} show that the AutoAttack method~\citep{croce2020reliable} tends to overestimate the robustness of diffusion models, primarily due to the presence of gradient obfuscation, which prevents the attack from effectively exploiting the true vulnerabilities of the model. To address this issue, recent studies \citep{chen2023robust,li2024adbm,liu2025towards} have adopted the gradient checkpointing technique to efficiently extract complete gradients throughout the diffusion process. Furthermore, Li et al.~\citep{li2024adbm} have further demonstrated that, compared to AutoAttack, the combination of PGD + EOT is more effective in evaluating the adaptive defense mechanisms of diffusion models. In line with these studies \citep{li2024adbm}, we employ the PGD200 + EOT20 configuration with $\ell_{\infty}(\epsilon=\frac{8}{255})$ and $\ell_{2}(\epsilon=1)$, utilizing the exact gradient computation method described in \citep{li2024adbm,liu2025towards} to ensure a more robust evaluation of defense performance in our experiments. Our threat model is purifier followed by a classifier. To ensure fair comparison with adversarial training methods, we use AutoAttack with full gradient settings \citep{chen2023robust} as the evaluation protocol, ensuring objectivity and comparability.

\begin{table*}[t]
\centering
\tiny
\caption{Comparison with adversarial training under Autoattack ($\epsilon = \frac{8}{255}$)}
\label{tab:adv_training}
\begin{tabular}{c|c|c|cc|cc|cc}
\midrule[1.2pt]
\multirow{2}{*}{Method} & \multirow{2}{*}{Extra data} & \multirow{2}{*}{Architecture} & \multicolumn{2}{c|}{CIFAR10} & \multicolumn{2}{c|}{CIFAR100} & \multicolumn{2}{c}{SVHN} \\ \cline{4-9} 
                        &                             &                               & Std Acc & $\ell_{\infty}$ & Std Acc & $\ell_{\infty}$ & Std Acc & $\ell_{\infty}$ \\ \midrule
Rebuffi et al.\citep{rebuffi2021fixing}  & \ding{51} & WRN-28-10 & 87.33 & 60.73 & 62.41 & 32.06 & 94.34 & 60.90 \\
Pang et al. \citep{pang2022robustness} & \ding{51} & WRN-28-10 & 88.10 & 61.51 & 62.08 & 31.40 & -- & -- \\
Wang et al.\citep{wang2023better}     & \ding{51} & WRN-28-10 & 91.12 & 63.35 & 68.06 & 35.65 & 95.19 & 61.85 \\ \midrule
\cellcolor{greyL}AMRM-Pure$_{\text{MAE}}$      & \cellcolor{greyL}\ding{55} & \cellcolor{greyL}WRN-28-10 & \cellcolor{greyL}88.57 & \cellcolor{greyL}40.65 & \cellcolor{greyL}65.34 & \cellcolor{greyL}16.77 & \cellcolor{greyL}94.54 & \cellcolor{greyL}47.03 \\
\cellcolor{greyL}AMRM-Pure$_{\text{MaskDiT}}$  & \cellcolor{greyL}\ding{55} & \cellcolor{greyL}WRN-28-10 & \cellcolor{greyL}92.03 & \cellcolor{greyL}64.97 & \cellcolor{greyL}\textbf{70.03} & \cellcolor{greyL}33.51 & \cellcolor{greyL}94.91 & \cellcolor{greyL}64.15 \\
\cellcolor{greyL}RAMRM-Pure$_{\text{MAE}}$     & \cellcolor{greyL}\ding{55} & \cellcolor{greyL}WRN-28-10 & \cellcolor{greyL}90.09 & \cellcolor{greyL}47.15 & \cellcolor{greyL}66.28 & \cellcolor{greyL}22.15 & \cellcolor{greyL}94.47 & \cellcolor{greyL}50.03 \\
\cellcolor{greyL}RAMRM-Pure$_{\text{MaskDiT}}$ & \cellcolor{greyL}\ding{55} & \cellcolor{greyL}WRN-28-10 & \cellcolor{greyL}\textbf{93.11} & \cellcolor{greyL}\textbf{75.83} & \cellcolor{greyL} 69.87 & \cellcolor{greyL}\textbf{40.13} & \cellcolor{greyL}\textbf{95.39} & \cellcolor{greyL}\textbf{66.12} \\ \midrule [1.2pt]
\end{tabular}
\end{table*}

\noindent \textbf{Evaluation Metrics.}
To evaluate the model's performance, we employ two metrics for classification: robust accuracy (\textbf{Robust Acc}) and standard accuracy (\textbf{Std Acc}), which are tested respectively on adversarial examples and clean samples. Due to the high computational cost of testing models with multiple attacks, we follow previous work \citep{nie2022diffusion,lin2024adversarial,li2024adbm} and randomly select 512 test samples from each testing dataset.  All results are from 5 different random seeds and we use its average values. The standard deviations of the experiments will be provided separately in the Supplementary Materials (not in the Appendix).


\subsection{Compare with the State-of-the-art}

We compare our results with state-of-the-art methods across four datasets: CIFAR-10, CIFAR-100, SVHN, and ImageNet. \textbf{Due to the space limitations, we provide detailed comparisons under the PGD200 + EoT20 attack for CIFAR-10, CIFAR-100, and SVHN in main paper. More experimental results, including performance on ImageNet, transferability of the fine-tuned purification model, defense against extra attacks, ablation studies, sensitivity analysis, and evaluations between different classifiers, are also presented in Appendix \ref{supp:exp}.}

\textbf{CIFAR-10.} Table \ref{cifar_10_full} highlights the performance of various purification methods on the CIFAR-10 dataset in terms of Std Acc and Robust Acc. RAMRM-Pure$_{\text{MaskDiT}}$ excels with 62. 13\% robust accuracy in $\ell{\infty}$ attacks, 73. 57\% in $\ell_{2}$ attacks, and strong standard accuracy of 93. 11\%. In contrast, traditional methods like the previous method~\citep{shi2021online} perform poorly, achieving only 4.56\% under $\ell_{\infty}$ attacks. In general, our method greatly improves adversarial robustness, with additional WideResNet-70-16 results provided in Supplementary \ref{sec:classifier}.

\textbf{CIFAR100.} Table \ref{tab:cifar100} summarizes the performance of various purification methods on CIFAR-100. Our proposed method achieves strong performance, with AMRM-Pure$_{\text{MaskDiT}}$ reaching the highest standard accuracy of 70.03\% and a robust accuracy of 24.39\% under $\ell_{\infty}$ attacks and 36.51\% under $\ell_{2}$ attacks. RAMRM-Pure$_{\text{MaskDiT}}$ further improves robustness, achieving the best $\ell_{\infty}$ robust accuracy of 29.91\% and $\ell_{2}$ robust accuracy of 43.27\%, outperforming other methods like Diffpure and COUP.

\textbf{SVHN.} Table \ref{tab:svhn} shows that ADBM achieves strong robust accuracy (47.90\% under $\ell_{\infty}$ attacks and 65.70\% under $\ell_{2}$) attacks but is outperformed by RAMRM-Pure$_{\text{MaskDiT}}$, which achieves the best robust accuracy (55.90\% and 70.18\%) with comparable standard accuracy. While RAMRM-Pure$_{\text{MAE}}$ leads in standard accuracy (95.39\%), its robustness is lower. Overall, our MaskDiT-based methods better balance clean and robust performance than ADBM.


\subsection{Comparison with adversarial training}

As shown in Table~\ref{tab:adv_training}, our methods achieve competitive or superior robustness compared to adversarial training baselines, \textbf{without using any extra data}. In contrast, prior works~\citep{rebuffi2021fixing, pang2022robustness, wang2023better} rely on \textbf{1M additional training samples}. Notably, \textbf{RAMRM-Pure$_{\text{MaskDiT}}$ achieves 75.83\% robust accuracy on CIFAR-10}, outperforming all baselines and highlighting the effectiveness of our data-free approach.

\subsection{Inference Time Comparison}

\begin{table}[h] 
\tiny
\centering
\caption{Inference time (s) consumption comparison across different defense models on CIFAR-10 and ImageNet datasets.}
\label{consume}
\begin{tabular}{ccc}
\midrule[1.2pt]
Defense Model & CIFAR10        & ImageNet       \\ \midrule
Diffpure \citep{nie2022diffusion}     & 12.39          & 81.54           \\
\midrule
\cellcolor{greyL}AMRM-Pure$_{\text{MAE}}$      & \cellcolor{greyL}18.25          & \cellcolor{greyL}\textbf{31.51} \\
\cellcolor{greyL}AMRM-Pure$_{\text{MaskDiT}}$      & \cellcolor{greyL}32.85          & \cellcolor{greyL}79.27          \\
\cellcolor{greyL}RAMRM-Pure$_{\text{MAE}}$     & \cellcolor{greyL}\textbf{11.77} & \cellcolor{greyL}27.38              \\
\cellcolor{greyL}RAMRM-Pure$_{\text{MaskDiT}}$    & \cellcolor{greyL}29.73          & \cellcolor{greyL}62.52     \\ \midrule[1.2pt]
\end{tabular}
\vspace{0cm}
\label{time_cons}
\end{table}

Table \ref{time_cons} compares the inference time between different defense models in CIFAR-10 and ImageNet. We calculate the run-time for all methods with a batch size of 32, and our experiments are conducted on an A40 GPU. For CIFAR-10, \textbf{RAMRM-Pure$_{\text{MAE}}$} achieves the fastest time (\textbf{11.77s}), followed by Diffpure (\textbf{12.39s}) and AMRM-Pure$_{\text{MAE}}$ (\textbf{18.25s}). On ImageNet, \textbf{AMRM-Pure$_{\text{MAE}}$} is the most efficient (\textbf{31.51s}), significantly outperforming Diffpure (\textbf{81.54s}). The results highlight that the DiffPure model has the advantage of inference time for smaller datasets, while our proposed model performs better on the metrics of inference time on the ImageNet dataset.

\section{CONCLUSION}
This paper reveals the vulnerability of AMRM to subtle adversarial attacks, caused by adversarial noises that disrupt semantic relations in image patches. To address this, we propose the AMRM-Pure pipeline, a denoising method using Attention Matrix Variation Minimization, which iteratively refines adversarial examples by minimizing reconstruction loss to converge to clean images. We further enhance AMRM-Pure with classifier loss, introducing RAMRM-Pure$_{\text{MAE}}$. We apply two AMRM (MAE and MaskDiT) to our method. Extensive experiments demonstrate that our methods achieve state-of-the-art performance on multiple benchmarks.

\section{ACKNOWLEDGEMENT}
The work was partially supported by the following: 
the Zhejiang Provincial Natural Science Foundation – Exploration Project under No. LMS26F020007,
the Wenzhou Applied Fundamental Research Program (Basic Research) under No. GG20250198,
the WKU 2026 International Frontier Interdisciplinary Research Institute Talent Program under No. WKUTP2026002,
the WKU 2025 International Collaborative Research Program under No. ICRPSP2025001.

\newpage
\bibliographystyle{abbrvnat}
\bibliography{reference}

@String(CVPR= {IEEE Conf. Comput. Vis. Pattern Recog.})

@String(NIPS= {Adv. Neural Inform. Process. Syst.})

@String(ICLR = {Int. Conf. Learn. Represent.})

@String(AAAI = {AAAI})

@String(CVPR  = {CVPR})

@String(NIPS  = {NeurIPS})

@String(ICLR  = {ICLR})

@article{lee2024visualizing,
  title={Visualizing the loss landscape of Self-supervised Vision Transformer},
  author={Lee, Youngwan and Willette, Jeffrey Ryan and Kim, Jonghee and Hwang, Sung Ju},
  journal={arXiv preprint arXiv:2405.18042},
  year={2024}
}

@article{bai2024diffusion,
  title={Diffusion Models Demand Contrastive Guidance for Adversarial Purification to Advance},
  author={Bai, Mingyuan and Huang, Wei and Li, Tenghui and Wang, Andong and Gao, Junbin and Caiafa, C{\'e}sar Federico and Zhao, Qibin},
  journal={Internal Conference on Machine Learning},
  year={2024},
  publisher={ECAI}
}

@inproceedings{madry2018towards,
  title={Towards Deep Learning Models Resistant to Adversarial Attacks},
  author={Madry, Aleksander and Makelov, Aleksandar and Schmidt, Ludwig and Tsipras, Dimitris and Vladu, Adrian},
  booktitle={International Conference on Learning Representations},
  year={2018}
}

@article{lin2024adversarial,
  title={Adversarial Training on Purification (AToP): Advancing Both Robustness and Generalization},
  author={Lin, Guang and Li, Chao and Zhang, Jianhai and Tanaka, Toshihisa and Zhao, Qibin},
  journal={International Conference on Learning Representations},
  year={2024}
}

@inproceedings{he2016deep,
  title={Deep residual learning for image recognition},
  author={He, Kaiming and Zhang, Xiangyu and Ren, Shaoqing and Sun, Jian},
  booktitle={Proceedings of the IEEE conference on computer vision and pattern recognition},
  pages={770--778},
  year={2016}
}

@inproceedings{moosavi2016deepfool,
  title={Deepfool: a simple and accurate method to fool deep neural networks},
  author={Moosavi-Dezfooli, Seyed-Mohsen and Fawzi, Alhussein and Frossard, Pascal},
  booktitle={Proceedings of the IEEE conference on computer vision and pattern recognition},
  pages={2574--2582},
  year={2016}
}

@inproceedings{zagoruyko2016wide,
  title={Wide Residual Networks},
  author={Zagoruyko, Sergey and Komodakis, Nikos},
  booktitle={British Machine Vision Conference 2016},
  year={2016}
}

@article{goodfellow2014explaining,
  title={Explaining and harnessing adversarial examples},
  author={Goodfellow, Ian J and Shlens, Jonathon and Szegedy, Christian},
  journal={International Conference
on Learning Representations},
  year={2015}
}

@article{kurakin2016adversarial,
  title={Adversarial machine learning at scale},
  author={Kurakin, Alexey and Goodfellow, Ian and Bengio, Samy},
  journal={n International Conference
on Learning Representations},
  year={2017}
}

@article{tramer2017ensemble,
  title={Ensemble adversarial training: Attacks and defenses},
  author={Tram{\`e}r, Florian and Kurakin, Alexey and Papernot, Nicolas and Goodfellow, Ian and Boneh, Dan and McDaniel, Patrick},
  journal={International Conference on Learning Representations},
  year={2018}
}

@inproceedings{zhang2019theoretically,
  title={Theoretically principled trade-off between robustness and accuracy},
  author={Zhang, Hongyang and Yu, Yaodong and Jiao, Jiantao and Xing, Eric and El Ghaoui, Laurent and Jordan, Michael},
  booktitle={International conference on machine learning},
  pages={7472--7482},
  year={2019}
}

@inproceedings{athalye2018synthesizing,
  title={Synthesizing robust adversarial examples},
  author={Athalye, Anish and Engstrom, Logan and Ilyas, Andrew and Kwok, Kevin},
  booktitle={International conference on machine learning},
  pages={284--293},
  year={2018}
}

@article{wong2020fast,
  title={Fast is better than free: Revisiting adversarial training},
  author={Wong, Eric and Rice, Leslie and Kolter, J Zico},
  journal={International Conference on Learning Representations},
  year={2020}
}

@inproceedings{liu2021using,
  title={Using single-step adversarial training to defend iterative adversarial examples},
  author={Liu, Guanxiong and Khalil, Issa and Khreishah, Abdallah},
  booktitle={Proceedings of the Eleventh ACM Conference on Data and Application Security and Privacy},
  pages={17--27},
  year={2021}
}

@article{gowal2021improving,
  title={Improving robustness using generated data},
  author={Gowal, Sven and Rebuffi, Sylvestre-Alvise and Wiles, Olivia and Stimberg, Florian and Calian, Dan Andrei and Mann, Timothy A},
  journal={Advances in Neural Information Processing Systems},
  volume={34},
  pages={4218--4233},
  year={2021}
}

@inproceedings{vivek2020single,
  title={Single-step adversarial training with dropout scheduling},
  author={Vivek, BS and Babu, R Venkatesh},
  booktitle={2020 IEEE/CVF Conference on Computer Vision and Pattern Recognition (CVPR)},
  pages={947--956},
  year={2020},
  organization={IEEE}
}

@article{wang2023better,
  title={Better diffusion models further improve adversarial training},
  author={Wang, Zekai and Pang, Tianyu and Du, Chao and Lin, Min and Liu, Weiwei and Yan, Shuicheng},
  journal={International Conference on Machine Learning},
  year={2023}
}

@article{rebuffi2021fixing,
  title={Fixing data augmentation to improve adversarial robustness},
  author={Rebuffi, Sylvestre-Alvise and Gowal, Sven and Calian, Dan A and Stimberg, Florian and Wiles, Olivia and Mann, Timothy},
  journal={arXiv preprint arXiv:2103.01946},
  year={2021}
}

@article{sehwag2021robust,
  title={Robust learning meets generative models: Can proxy distributions improve adversarial robustness?},
  author={Sehwag, Vikash and Mahloujifar, Saeed and Handina, Tinashe and Dai, Sihui and Xiang, Chong and Chiang, Mung and Mittal, Prateek},
  journal={ International Conference on Learning Representations},
  year={2021}
}

@article{samangouei2018defense,
  title={Defense-gan: Protecting classifiers against adversarial attacks using generative models},
  author={Samangouei, Pouya and Kabkab, Maya and Chellappa, Rama},
  journal={ International Conference on Learning Representations},
  year={2018}
}

@article{song2017pixeldefend,
  title={Pixeldefend: Leveraging generative models to understand and defend against adversarial examples},
  author={Song, Yang and Kim, Taesup and Nowozin, Sebastian and Ermon, Stefano and Kushman, Nate},
  journal={International Conference on Learning
Representations},
  year={2017}
}

@inproceedings{yoon2021adversarial,
  title={Adversarial purification with score-based generative models},
  author={Yoon, Jongmin and Hwang, Sung Ju and Lee, Juho},
  booktitle={International Conference on Machine Learning},
  pages={12062--12072},
  year={2021},
  organization={PMLR}
}

@article{nie2022diffusion,
  title={Diffusion models for adversarial purification},
  author={Nie, Weili and Guo, Brandon and Huang, Yujia and Xiao, Chaowei and Vahdat, Arash and Anandkumar, Anima},
  journal={International Conference on Machine Learning},
  year={2022}
}

@article{song2020score,
  title={Score-based generative modeling through stochastic differential equations},
  author={Song, Yang and Sohl-Dickstein, Jascha and Kingma, Diederik P and Kumar, Abhishek and Ermon, Stefano and Poole, Ben},
  journal={International Conference on Learning Representations},
  year={2021}
}

@article{shi2021online,
  title={Online adversarial purification based on self-supervision},
  author={Shi, Changhao and Holtz, Chester and Mishne, Gal},
  journal={ International Conference on Learning Representations},
  year={2021}
}

@inproceedings{pang2022robustness,
  title={Robustness and accuracy could be reconcilable by (proper) definition},
  author={Pang, Tianyu and Lin, Min and Yang, Xiao and Zhu, Jun and Yan, Shuicheng},
  booktitle={International Conference on Machine Learning},
  pages={17258--17277},
  year={2022},
  organization={PMLR}
}

@article{12,
  author={Carlini, N. and Wagner, D. },
  journal={in 2017 IEEE Symposium on Security and Privacy (SP), pp. }, 
  title={Towards evaluating the robustness of neural networks. In 2017 IEEE Symposium on Security and Privacy (SP)}, 
  year={2017},
   pages={39–57}
  }

@article{13,
author = {Dawn Song and Kevin Eykholt and Ivan Evtimov and Earlence Fernandes and Bo Li and Amir Rahmati and Florian Tram{\`e}r and Atul Prakash and Tadayoshi Kohno},
title = {Physical Adversarial Examples for Object Detectors},
journal = {12th USENIX Workshop on Offensive Technologies (WOOT 18)},
year = {2018},
address = {Baltimore, MD},
publisher = {USENIX Association},
month = aug,
}

@article{14,
  title={Adversarial examples for semantic image segmentation},
  author={Fischer, Volker and Kumar, Mummadi Chaithanya and Metzen, Jan Hendrik and Brox, Thomas},
  journal={International Conference on Computer Vision},
  year={2017}
}

@article{15,
  title={A unified gradient regularization family for adversarial examples},
  author={Lyu, Chunchuan and Huang, Kaizhu and Liang, Hai-Ning},
  journal={2015 IEEE international conference on data mining},
  pages={301--309},
  year={2015},
}

@inproceedings{he2022masked,
  title={Masked autoencoders are scalable vision learners},
  author={He, Kaiming and Chen, Xinlei and Xie, Saining and Li, Yanghao and Doll{\'a}r, Piotr and Girshick, Ross},
  booktitle={Proceedings of the IEEE/CVF conference on computer vision and pattern recognition},
  pages={16000--16009},
  year={2022}
}

@article{choromanski2020rethinking,
  title={Rethinking attention with performers},
  author={Choromanski, Krzysztof and Likhosherstov, Valerii and Dohan, David and Song, Xingyou and Gane, Andreea and Sarlos, Tamas and Hawkins, Peter and Davis, Jared and Mohiuddin, Afroz and Kaiser, Lukasz and others},
  journal={International Conference on Learning Representations},
  year={2021}
}

@article{krizhevsky2009learning,
  title={Learning multiple layers of features from tiny images},
  author={Krizhevsky, Alex and Hinton, Geoffrey and others},
  year={2009},
  publisher={Toronto, ON, Canada}
}

@inproceedings{netzer2011reading,
  title={Reading digits in natural images with unsupervised feature learning},
  author={Netzer, Yuval and Wang, Tao and Coates, Adam and Bissacco, Alessandro and Wu, Baolin and Ng, Andrew Y and others},
  booktitle={NIPS workshop on deep learning and unsupervised feature learning},
  volume={2011},
  number={5},
  pages={7},
  year={2011},
  organization={Granada, Spain}
}

@inproceedings{deng2009imagenet,
  title={Imagenet: A large-scale hierarchical image database},
  author={Deng, Jia and Dong, Wei and Socher, Richard and Li, Li-Jia and Li, Kai and Fei-Fei, Li},
  booktitle={2009 IEEE conference on computer vision and pattern recognition},
  pages={248--255},
  year={2009},
  organization={Ieee}
}

@article{zhang2022mask,
  title={How mask matters: Towards theoretical understandings of masked autoencoders},
  author={Zhang, Qi and Wang, Yifei and Wang, Yisen},
  journal={Advances in Neural Information Processing Systems},
  volume={35},
  pages={27127--27139},
  year={2022}
}

@article{cao2022understand,
  title={How to understand masked autoencoders},
  author={Cao, Shuhao and Xu, Peng and Clifton, David A},
  journal={arXiv preprint arXiv:2202.03670},
  year={2022}
}

@inproceedings{zhou2023eliminating,
  title={Eliminating adversarial noise via information discard and robust representation restoration},
  author={Zhou, Dawei and Chen, Yukun and Wang, Nannan and Liu, Decheng and Gao, Xinbo and Liu, Tongliang},
  booktitle={International Conference on Machine Learning},
  pages={42517--42530},
  year={2023},
  organization={PMLR}
}

@article{wu2022denoising,
  title={Denoising masked autoencoders help robust classification},
  author={Wu, QuanLin and Ye, Hang and Gu, Yuntian and Zhang, Huishuai and Wang, Liwei and He, Di},
  journal={ICLR},
  year={2022}
}

@article{you2023beyond,
  title={Beyond pretrained features: noisy image modeling provides adversarial defense},
  author={You, Zunzhi and Liu, Daochang and Han, Bohyung and Xu, Chang},
  journal={Advances in Neural Information Processing Systems},
  volume={36},
  year={2023}
}

@incollection{zhang2024classifier,
  title={Classifier Guidance Enhances Diffusion-based Adversarial Purification by Preserving Predictive Information},
  author={Zhang, Mingkun and Li, Jianing and Chen, Wei and Guo, Jiafeng and Cheng, Xueqi},
  booktitle={ECAI 2024},
  pages={2234--2241},
  year={2024},
  publisher={IOS Press}
}

@article{chen2023robust,
  title={Robust classification via a single diffusion model},
  author={Chen, Huanran and Dong, Yinpeng and Wang, Zhengyi and Yang, Xiao and Duan, Chengqi and Su, Hang and Zhu, Jun},
  journal={ICML},
  year={2024}
}

@article{li2024adbm,
  title={ADBM: Adversarial diffusion bridge model for reliable adversarial purification},
  author={Li, Xiao and Sun, Wenxuan and Chen, Huanran and Li, Qiongxiu and Liu, Yining and He, Yingzhe and Shi, Jie and Hu, Xiaolin},
  journal={ICLR},
  year={2025}
}

@inproceedings{
liu2025towards,
title={Towards Understanding the Robustness of Diffusion-Based Purification: A Stochastic Perspective},
author = {Yiming Liu and Kezhao Liu and Yao Xiao and Ziyi Dong and Xiaogang Xu and Pengxu Wei and Liang Lin},
booktitle={The Thirteenth International Conference on Learning Representations},
year={2025},
}

@article{zhang2023enhancing,
  title={Enhancing adversarial robustness via score-based optimization},
  author={Zhang, Boya and Luo, Weijian and Zhang, Zhihua},
  journal={Advances in Neural Information Processing Systems},
  volume={36},
  pages={51810--51829},
  year={2023}
}

@inproceedings{andriushchenko2020square,
  title={Square attack: a query-efficient black-box adversarial attack via random search},
  author={Andriushchenko, Maksym and Croce, Francesco and Flammarion, Nicolas and Hein, Matthias},
  booktitle={European conference on computer vision},
  pages={484--501},
  year={2020},
  organization={Springer}
}

@inproceedings{croce2020minimally,
  title={Minimally distorted adversarial examples with a fast adaptive boundary attack},
  author={Croce, Francesco and Hein, Matthias},
  booktitle={International Conference on Machine Learning},
  pages={2196--2205},
  year={2020},
  organization={PMLR}
}

@inproceedings{croce2020reliable,
  title={Reliable evaluation of adversarial robustness with an ensemble of diverse parameter-free attacks},
  author={Croce, Francesco and Hein, Matthias},
  booktitle={International conference on machine learning},
  pages={2206--2216},
  year={2020},
  organization={PMLR}
}

@inproceedings{chen2020rays,
  title={Rays: A ray searching method for hard-label adversarial attack},
  author={Chen, Jinghui and Gu, Quanquan},
  booktitle={Proceedings of the 26th ACM SIGKDD International Conference on Knowledge Discovery \& Data Mining},
  pages={1739--1747},
  year={2020}
}

@article{
zheng2024fast,
title={Fast Training of Diffusion Models with Masked Transformers},
author={Hongkai Zheng and Weili Nie and Arash Vahdat and Anima Anandkumar},
journal={Transactions on Machine Learning Research},
issn={2835-8856},
year={2024},

}

@article{vaswani2017attention,
  title={Attention is all you need},
  author={Vaswani, Ashish and Shazeer, Noam and Parmar, Niki and Uszkoreit, Jakob and Jones, Llion and Gomez, Aidan N and Kaiser, {\L}ukasz and Polosukhin, Illia},
  journal={Advances in neural information processing systems},
  volume={30},
  year={2017}
}

@inproceedings{pathak2016context,
  title={Context encoders: Feature learning by inpainting},
  author={Pathak, Deepak and Krahenbuhl, Philipp and Donahue, Jeff and Darrell, Trevor and Efros, Alexei A},
  booktitle={Proceedings of the IEEE conference on computer vision and pattern recognition},
  pages={2536--2544},
  year={2016}
}

@article{nguyen2022fourierformer,
  title={Fourierformer: Transformer meets generalized fourier integral theorem},
  author={Nguyen, Tan and Pham, Minh and Nguyen, Tam and Nguyen, Khai and Osher, Stanley and Ho, Nhat},
  journal={Advances in Neural Information Processing Systems},
  volume={35},
  pages={29319--29335},
  year={2022}
}

@inproceedings{mei2025efficient,
  title={Efficient image-to-image diffusion classifier for adversarial robustness},
  author={Mei, Hefei and Dong, Minjing and Xu, Chang},
  booktitle={Proceedings of the AAAI Conference on Artificial Intelligence},
  volume={39},
  number={6},
  pages={6081--6089},
  year={2025}
}

@article{dou2024improving,
  title={Improving Robust Generalization with Diverging Spanned Latent Space},
  author={Dou, Owen and Gao, Zhiqiang and Shen, Hangchi and Yuan, Ziling and Zhang, Shufei and Huang, Kaizhu},
  journal={Transactions on Machine Learning Research},
  year={2024}
}








\section*{Checklist}

\begin{enumerate}

  \item For all models and algorithms presented, check if you include:
  \begin{enumerate}
    \item A clear description of the mathematical setting, assumptions, algorithm, and/or model. [Yes]
    \item An analysis of the properties and complexity (time, space, sample size) of any algorithm. [Yes]
    \item (Optional) Anonymized source code, with specification of all dependencies, including external libraries. [Yes]
  \end{enumerate}

  \item For any theoretical claim, check if you include:
  \begin{enumerate}
    \item Statements of the full set of assumptions of all theoretical results. [Yes]
    \item Complete proofs of all theoretical results. [Yes]
    \item Clear explanations of any assumptions. [Yes]     
  \end{enumerate}

  \item For all figures and tables that present empirical results, check if you include:
  \begin{enumerate}
    \item The code, data, and instructions needed to reproduce the main experimental results (either in the supplemental material or as a URL). [Yes]
    \item All the training details (e.g., data splits, hyperparameters, how they were chosen). [Yes]
    \item A clear definition of the specific measure or statistics and error bars (e.g., with respect to the random seed after running experiments multiple times). [Yes]
    \item A description of the computing infrastructure used. (e.g., type of GPUs, internal cluster, or cloud provider). [Yes]
  \end{enumerate}

  \item If you are using existing assets (e.g., code, data, models) or curating/releasing new assets, check if you include:
  \begin{enumerate}
    \item Citations of the creator If your work uses existing assets. [Yes]
    \item The license information of the assets, if applicable. [Yes]
    \item New assets either in the supplemental material or as a URL, if applicable. [Yes/No/Not Applicable]
    \item Information about consent from data providers/curators. [Yes]
    \item Discussion of sensible content if applicable, e.g., personally identifiable information or offensive content. [Yes]
  \end{enumerate}

  \item If you used crowdsourcing or conducted research with human subjects, check if you include:
  \begin{enumerate}
    \item The full text of instructions given to participants and screenshots. [Not Applicable]
    \item Descriptions of potential participant risks, with links to Institutional Review Board (IRB) approvals if applicable. [Not Applicable]
    \item The estimated hourly wage paid to participants and the total amount spent on participant compensation. [Not Applicable]
  \end{enumerate}

\end{enumerate}

\clearpage
\appendix
\thispagestyle{empty}

\onecolumn
\aistatstitle{AMRM-Pure: Semantic-Preserving Adversarial Purification: \\
Supplementary Materials}

\section{Limitation}

One limitation of AMRM-Pure$_{\text{MAE}}$ lies in its linear memory consumption with respect to batch size, which may restrict scalability under limited GPU resources.

\section{Broader impact}
We are the first to systematically explore the relationship between adversarial noise and inter-patch semantic information. While existing defense methods primarily focus on suppressing pixel-level perturbations or enhancing model robustness structurally, our work takes a novel perspective by analyzing reconstruction consistency and semantic alignment. We reveal how adversarial perturbations disrupt inter-patch semantic relations and propose a reconstruction paradigm that restores this consistency. This new angle provides a valuable direction for future adversarial defense research and advances the theoretical and practical understanding of robustness from a structure-aware perspective.


\section{Supplement Experiment}
\label{supp:exp}
We have enhanced this section with additional experiments to provide a more comprehensive evaluation of our work. Specifically, we present ImageNet~\citep{deng2009imagenet} results under PGD200 + EoT20~\citep{madry2018towards,athalye2018synthesizing} with the perturbation budgets $\epsilon=\frac{4}{255}$ for $\ell_{\infty}$ attack and $\epsilon=0.5$ for $\ell_{2}$ attack, using ResNet-101~\citep{he2016deep} as the classifier. We also performed ablation studies on CIFAR-10~\citep{krizhevsky2009learning}, CIFAR-100~\citep{krizhevsky2009learning}, and SVHN~\citep{netzer2011reading} using different backbones and evaluated diverse attack scenarios.

\subsection{Performance on ImageNet.}
Table \ref{image_net} presents the standard accuracy and robust accuracy of different purification methods on the ImageNet dataset under the $\ell_{\infty}$ attack. As shown in the table, the ADDT method achieves the highest standard accuracy at \textbf{80.20\%}, slightly outperforming the other methods. However, in terms of robustness, RAMRM-Pure$_{\text{MaskDiT}}$ stands out with a robust accuracy of \textbf{36.87\%}, surpassing all other methods, including ADDT. In contrast, AMRM-Pure$_{\text{MAE}}$ and AMRM-Pure$_{\text{MaskDiT}}$ demonstrate relatively lower robustness, achieving 24.75\% and 32.29\%, respectively. 

\begin{table}[!h]
\centering
\caption{Clean and robust accuracy (\%) with $\ell_{\infty}$ and $\ell_2$ attack on ImageNet obtained
by different purification methods.}
\begin{tabular}{c|ccc}
\midrule[1.2pt]
\multirow{2}{*}{Defense model} & \multirow{2}{*}{Std Acc} & \multicolumn{2}{c}{Robust Acc}   \\ \cline{3-4} 
                               &                          & $\ell_{\infty}$ & $\ell_{2}$     \\ \midrule
Diffpure \citep{nie2022diffusion}                       & 77.51                    & 30.15           & 44.15          \\
ADDT  \citep{liu2021using}                         & \textbf{80.20}           & 35.83           & -          \\ \midrule
\cellcolor{greyL}AMRM-Pure$_{\text{MAE}}$                       & \cellcolor{greyL}67.53                    & \cellcolor{greyL}24.75           & \cellcolor{greyL}35.95          \\
\cellcolor{greyL}AMRM-Pure$_{\text{MaskDiT}}$                   & \cellcolor{greyL}75.52                    & \cellcolor{greyL}32.29           & \cellcolor{greyL}45.57          \\
\cellcolor{greyL}RAMRM-Pure$_{\text{MAE}}$                      & \cellcolor{greyL}78.85                    & \cellcolor{greyL}29.48           & \cellcolor{greyL}42.25          \\
\cellcolor{greyL}RAMRM-Pure$_{\text{MaskDiT}}$                  & \cellcolor{greyL}79.52                    & \cellcolor{greyL}\textbf{36.87}  & \cellcolor{greyL}\textbf{51.17} \\ \midrule[1.2pt]
\end{tabular}
\label{image_net}
\end{table}



\subsection{ Transferability of finetuned purification on new classifiers}

We fine-tune the RAMRM-Pure$_{\text{MAE}}$/RAMRM-Pure$_{\text{MaskDiT}}$ model based on WideResNet-28-10~\citep{zagoruyko2016wide} and replace it with different classifiers for testing experiments. The WideResNet-70-16~\citep{zagoruyko2016wide} and ResNet-50~\citep{he2016deep} are selected for testing process to observe the transferability of our proposed method across different classifiers.

\begin{table}[!h]
\centering
    \caption{Robust accuracy (\%) of different purification methods against $\ell_{\infty}(\epsilon=\frac{8}{255})$ and $\ell_{2}(\epsilon=1)$ adversarial attacks across two classifiers: WideResNet-70-16 and ResNet-50. Here, our method are derived by fine-tuning on the WideResNet-28-10 classifier.}
\begin{tabular}{c|cc|cc}
\midrule[1.2pt]
\multirow{2}{*}{Classifier}  & \multicolumn{2}{c|}{WideResNet-70-16} & \multicolumn{2}{c}{ResNet-50}    \\ \cline{2-5} 
                             & $\ell_{\infty}$    & $\ell_{2}$        & $\ell_{\infty}$ & $\ell_{2}$     \\ \midrule
Diffpure \citep{nie2022diffusion}                    & 42.20              & 60.80             & 38.02           & 54.74          \\ \midrule
\cellcolor{greyL}RAMRM-Pure$_{\text{MAE}}$     & \cellcolor{greyL}44.27              & \cellcolor{greyL}62.42             & \cellcolor{greyL}43.72           & \cellcolor{greyL}60.08          \\
\cellcolor{greyL}RAMRM-Pure$_{\text{MaskDiT}}$ & \cellcolor{greyL}\textbf{58.37}     & \cellcolor{greyL}\textbf{68.55}    & \cellcolor{greyL}\textbf{54.22}  & \cellcolor{greyL}\textbf{67.15} \\ \midrule[1.2pt]
\end{tabular}
\label{transfer}
\end{table}

The results in Table \ref{transfer} highlight the superior performance of our method, particularly RAMRM-Pure$_{\text{MaskDiT}}$, which achieves the highest robust accuracy across both classifiers and attack norms. Notably, the fine-tuned models, initially trained on WideResNet-28-10, demonstrate strong transferability when applied to WideResNet-70-16 and ResNet-50 without the need for retraining. This finding underscores the practicality and scalability of our method, as its fine-tuned models can be seamlessly adapted to new classifiers, providing an efficient and robust defense against adversarial attacks.

\subsection{Performance on unseen threats}

\begin{table}[!h]
\centering
\caption{Robust accuracy (\%) against unseen threats with the setting of $\ell_{1}(\epsilon=12)$ and $\ell_{2}(\epsilon=1)$.}
\begin{tabular}{c|cccccc}
\midrule[1.2pt]
\multirow{2}{*}{Defense model} & \multicolumn{2}{c}{CIFAR-10}    & \multicolumn{2}{c}{CIFAR-100}   & \multicolumn{2}{c}{SVHN} \\ \cline{2-7} 
                               & $\ell_{1}$     & $\ell_{2}$     & $\ell_{1}$     & $\ell_{2}$     & $\ell_{1}$  & $\ell_{2}$ \\ \midrule
Diffpure \citep{nie2022diffusion}                       & 44.30          & 60.80          & 13.51          & 27.53          & 46.10       & 63.30      \\
ADBM \citep{li2024adbm}                           & 49.60          & 63.30          & -              & -              & 51.20       & 65.70      \\ \midrule
\cellcolor{greyL}RAMRM-Pure$_{\text{MAE}}$                      & \cellcolor{greyL}44.41          & \cellcolor{greyL}60.72          & \cellcolor{greyL}12.97          & \cellcolor{greyL}29.58          & \cellcolor{greyL}47.09       & \cellcolor{greyL}60.51      \\
\cellcolor{greyL}RAMRM-Pure$_{\text{MaskDiT}}$                  & \cellcolor{greyL}\textbf{65.11} & \cellcolor{greyL}\textbf{73.57} & \cellcolor{greyL}\textbf{41.15} & \cellcolor{greyL}\textbf{43.27} & \cellcolor{greyL}\textbf{55.53}       & \cellcolor{greyL}\textbf{70.18}      \\ \midrule[1.2pt]
\end{tabular}
\label{unseen}
\end{table}

For the three methods, ADBM, RAMRM-Pure$_{\text{MAE}}$, and RAMRM-Pure$_{\text{MaskDiT}}$, all of which are fine-tuned under the $\ell_{\infty}$ norm, the $\ell_{1}$ and $\ell_{2}$ norms are considered as unseen threats. To verify the robustness of the proposed method, we will now conduct testing under these unseen threats. For a fair comparison on the CIFAR-10 dataset, we employ the WideResNet-70-16 architecture, and we use the WideResNet-28-10 architecture on the CIFAR-100 and SVHN datasets.

Table \ref{unseen} presents the robust accuracy (\%) of different defense models against unseen threats ($\ell_{1}$ and $\ell_{2}$ attacks) on the CIFAR-10, CIFAR-100, and SVHN datasets. As a baseline method, DiffPure \citep{nie2022diffusion} performs moderately on CIFAR-10 and SVHN but poorly on CIFAR-100, especially under $\ell_{1}$ attacks (13.51\%). ADBM outperforms DiffPure on CIFAR-10 and SVHN, but no data is provided for CIFAR-100, suggesting potential limitations or untested performance on this dataset. Our RAMRM-Pure$_{\text{MAE}}$ method slightly outperforms DiffPure on CIFAR-10 and SVHN but underperforms on CIFAR-100 (12.97\% vs. 13.51\%), indicating some limitations on more complex datasets. In contrast, RAMRM-Pure$_{\text{MaskDiT}}$ significantly outperforms all other methods across all datasets and attack types. On CIFAR-10, RAMRM-Pure$_{\text{MaskDiT}}$ achieves accuracies of 65.11\% and 73.57\% under $\ell_{1}$ and $\ell_{2}$ attacks, respectively, far surpassing other methods. On CIFAR-100, although its performance under $\ell_{2}$ attacks is slightly lower than under $\ell_{1}$, it still outperforms other methods. On the SVHN dataset, RAMRM-Pure$_{\text{MaskDiT}}$ also demonstrates considerable robustness performance, particularly under $\ell_{2}$ attacks (70.18\%). Overall, RAMRM-Pure$_{\text{MaskDiT}}$ exhibits the strongest robustness against unseen threats, especially on CIFAR-10 and CIFAR-100, showcasing its superior generalization and defense capabilities, while RAMRM-Pure$_{\text{MAE}}$, though slightly less effective, still outperforms baseline methods in certain scenarios.

\subsection{Defense against adaptive attacks}
Table \ref{extra_attacks} presents the robust accuracy (\%) of various defense methods under different adversarial attacks in the adaptive \(\ell_{2}(\epsilon=1)\)-norm setting on the CIFAR-10 dataset. The evaluated adaptive attacks include C\&W~\citep{12}+EOT~\citep{athalye2018synthesizing}, DeepFool~\citep{moosavi2016deepfool}+EOT, AutoAttack~\citep{croce2020reliable}+EOT, and PGD~\citep{madry2018towards}+EOT. Among the methods, Diffpure and ADBM represent baseline defense approaches, with ADBM generally outperforming Diffpure across all attacks. For instance, ADBM achieves 78.40\% robust accuracy against C\&W+EOT compared to Diffpure's 74.80\%. The pure methods (non-adversarial training approaches) show varying performance: AMRM-Pure$_{\text{MAE}}$ exhibits the lowest robust accuracy across all attacks, while AMRM-Pure$_{\text{MaskDiT}}$ demonstrates stronger performance, particularly against DeepFool+EOT (82.8\%). RAMRM-Pure$_{\text{MAE}}$ shows moderate results, and RAMRM-Pure$_{\text{MaskDiT}}$ consistently outperforms all other methods, achieving the highest robust accuracy against every attack, with 80.58\% for C\&W+EOT, 86.11\% for DeepFool+EOT, 73.59\% for AutoAttack+EOT, and 69.57\% for PGD+EOT. This indicates that RAMRM-Pure$_{\text{MaskDiT}}$ is the most effective defense method in this setting, offering superior robustness across diverse adversarial attacks.
\begin{table}[!h]
\caption{Robust Accuracy (\%) of various defense methods under different attacks in the $\ell_{2}(\epsilon=1)$-norm setting using the exact gradient with WideResNet-70-16 on CIFAR-10.}
\centering
\begin{tabular}{ccccc}
\midrule[1.2pt]
Method        & C\&W+EOT       & DeepFool+EOT   & AutoAttack+EOT & PGD+EOT        \\ \midrule
Diffpure\citep{nie2022diffusion}      & 74.80          & 78.40          & 63.90          & 60.80          \\
ADBM\citep{li2024adbm}          & 78.40          & 84.30          & 66.80          & 66.30          \\ \midrule
\cellcolor{greyL}AMRM-Pure$_{\text{MAE}}$      & \cellcolor{greyL}62.54          & \cellcolor{greyL}65.15          & \cellcolor{greyL}52.75          & \cellcolor{greyL}53.50          \\
\cellcolor{greyL}AMRM-Pure$_{\text{MaskDiT}}$  & \cellcolor{greyL}79.15          & \cellcolor{greyL}82.80           & \cellcolor{greyL}66.43          & \cellcolor{greyL}64.53          \\
\cellcolor{greyL}RAMRM-Pure$_{\text{MAE}}$     & \cellcolor{greyL}72.20          & \cellcolor{greyL}72.29          & \cellcolor{greyL}59.11          & \cellcolor{greyL}60.72          \\
\cellcolor{greyL}RAMRM-Pure$_{\text{MaskDiT}}$ & \cellcolor{greyL}\textbf{80.58} & \cellcolor{greyL}\textbf{86.11} & \cellcolor{greyL}\textbf{73.59} & \cellcolor{greyL}\textbf{69.57} \\ \midrule[1.2pt]
\end{tabular}
\label{extra_attacks}
\end{table}


\subsection{Extra experiments on different classifier}
\label{sec:classifier}
Table \ref{WRN70_} shows our standard and robust accuracy using WideResNet-70-16 under CIFAR-10 and SVHN. Compared with WideResNet-28-10, it shows better results. It means the overparameterization contributes model's robustness. Among all the methods, the effectiveness of the RAMRM-Pure$_{\text{MaskDiT}}$ method is most notable.
\begin{table}[!h]
\centering
\caption{Performance of standard accuracy and robust accuracy (\%) using WideResNet-70-16.}
\begin{tabular}{c|c|ccc|ccc}
\midrule[1.2pt]
\multirow{2}{*}{Method} & \multirow{2}{*}{Architecture} & \multicolumn{3}{c|}{CIFAR10}                      & \multicolumn{3}{c}{SVHN}                          \\ \cline{3-8} 
                        &                               & Std Acc        & $\ell_{\infty}$ & $\ell_{2}$     & Std Acc        & $\ell_{\infty}$ & $\ell_{2}$     \\ \midrule
\cellcolor{greyL}AMRM-Pure$_{\text{MAE}}$                & \cellcolor{greyL}WideResNet-70-16             & \cellcolor{greyL}89.66          & \cellcolor{greyL}42.21           & \cellcolor{greyL}56.77          & \cellcolor{greyL}94.97          & \cellcolor{greyL}17.11           & \cellcolor{greyL}32.75          \\
\cellcolor{greyL}AMRM-Pure$_{\text{MaskDiT}}$            & \cellcolor{greyL}WideResNet-70-16             & \cellcolor{greyL}94.91          & \cellcolor{greyL}52.07           & \cellcolor{greyL}66.55          & \cellcolor{greyL}94.93          & \cellcolor{greyL}46.03           & \cellcolor{greyL}63.77          \\
\cellcolor{greyL}RAMRM-Pure$_{\text{MAE}}$               & \cellcolor{greyL}WideResNet-70-16             & \cellcolor{greyL}91.07          & \cellcolor{greyL}47.92           & \cellcolor{greyL}60.95          & \cellcolor{greyL}94.78          & \cellcolor{greyL}40.79           & \cellcolor{greyL}60.51          \\
\cellcolor{greyL}RAMRM-Pure$_{\text{MaskDiT}}$           & \cellcolor{greyL}WideResNet-70-16             & \cellcolor{greyL}\textbf{93.85} & \cellcolor{greyL}\textbf{63.94}  & \cellcolor{greyL}\textbf{75.50} & \cellcolor{greyL}\textbf{95.91} & \cellcolor{greyL}\textbf{56.02}  & \cellcolor{greyL}\textbf{69.18} \\ \midrule[1.2pt]
\end{tabular}
\label{WRN70_}
\end{table}


\subsection{Performance on Black-box Attack}

To evaluate the effectiveness of against black-box attacks, we adopt three black-box attack methods: FAB \citep{croce2020minimally}, Square \citep{andriushchenko2020square}, and Rays \citep{chen2020rays} on CIFAR-10 and SVHN. The black-box scenario implies that the attacker has no knowledge of the defense method.
Table \ref{black_tab} shows the robustness of various methods against black-box $\ell_{\infty}$ attacks with the perturbation budget $\epsilon = \frac{8}{255}$ using WideResNet-28-10. RMaskDiT still achieves the best results.

\begin{table}[!h]
\footnotesize
\setlength{\tabcolsep}{2pt} 
\caption{Robust accuracy (\% ) against different black-box attacks $\ell_{\infty}(\epsilon=\frac{8}{255})$ with WideResNet-28-10. The ``Vanilla" setting represents the model trained on clean datasets without any defense.}
\centering
\begin{tabular}{c|c|clcclcccc}
\midrule[1.2pt]
\multirow{3}{*}{Method} & \multirow{3}{*}{Architecture} & \multicolumn{4}{c}{CIFAR-10}                                                                           &  & \multicolumn{4}{c}{SVHN}                                                                \\ \cline{3-6} \cline{8-11} 
                        &                               & \multirow{2}{*}{Std Acc} & \multicolumn{3}{c}{Robust Acc}                                              &  & \multirow{2}{*}{Std Acc} & \multicolumn{3}{c}{Robust Acc}                               \\ \cline{4-6} \cline{9-11} 
                        &                               &                          & Square                & FAB                      & RayS                     &  &                          & \multicolumn{1}{l}{Square} & FAB            & RayS           \\ \cline{1-6} \cline{8-11} 
Vanilla                 & WideResNet-28-10             & 96.75                    & 19.15                 & \multicolumn{1}{l}{0.00} & \multicolumn{1}{l}{1.23} &  & 98.11                    & 9.08                       & 14.78          & 16.89          \\
Diffpure \citep{nie2022diffusion}                & WideResNet-28-10             & 89.15                    & 89.15                 & 88.29                    & 90.51                    &  & 93.93                    & 92.15                      & 93.13          & 92.97          \\
ADBM \citep{li2024adbm}                    & WideResNet-28-10             & -                        & \multicolumn{1}{c}{-} & -                        & -                        &  & 93.49                    & 93.32                      & 92.98          & 93.16          \\ \cline{1-6} \cline{8-11} 
\cellcolor{greyL}AMRM-Pure$_{\text{MAE}}$                & \cellcolor{greyL}WideResNet-28-10             & \cellcolor{greyL}88.57                    & \cellcolor{greyL}78.59                 & \cellcolor{greyL}76.43                    & \cellcolor{greyL}77.29                    &  & \cellcolor{greyL}94.54                    & \cellcolor{greyL}92.57                      & \cellcolor{greyL}93.36          & \cellcolor{greyL}93.41          \\
\cellcolor{greyL}AMRM-Pure$_{\text{MaskDiT}}$            & \cellcolor{greyL}WideResNet-28-10             & \cellcolor{greyL}92.03                    & \cellcolor{greyL}90.96                 & \cellcolor{greyL}92.25                    & \cellcolor{greyL}\textbf{93.39}           &  & \cellcolor{greyL}94.91                    & \cellcolor{greyL}92.80                      & \cellcolor{greyL}92.59 & \cellcolor{greyL}92.73          \\
\cellcolor{greyL}RAMRM-Pure$_{\text{MAE}}$               & \cellcolor{greyL}WideResNet-28-10             & \cellcolor{greyL}90.09                    & \cellcolor{greyL}90.25                 & \cellcolor{greyL}89.15                    & \cellcolor{greyL}92.31                    &  & \cellcolor{greyL}94.47                    & \cellcolor{greyL}92.73                      & \cellcolor{greyL}93.27          & \cellcolor{greyL}93.36          \\
\cellcolor{greyL}RAMRM-Pure$_{\text{MaskDiT}}$           & \cellcolor{greyL}WideResNet-28-10             & \cellcolor{greyL} \textbf{93.11}                    & \cellcolor{greyL}\textbf{93.27}        & \cellcolor{greyL}\textbf{93.38}           & \cellcolor{greyL}93.03                    &  & \cellcolor{greyL} \textbf{95.39}                    & \cellcolor{greyL}\textbf{94.15}             & \cellcolor{greyL}\textbf{94.18} & \cellcolor{greyL}\textbf{94.88} \\ \midrule[1.2pt]
\end{tabular}
\label{black_tab}
\end{table}



\subsection{Robust under BPDA attack}

We evaluate the robustness of our model under a strong white-box attack setting using BPDA combined with EoT set to 20. The results show as follow:

\begin{table}[h]
\centering
\caption{Clean and robust accuracy (\%) under BPDA attack ($\epsilon = \frac{8}{255}$) on CIFAR-10.}
\begin{tabular}{c|c|c|c}
\hline
\multirow{1}{*}{Method}  & \multirow{1}{*}{Architecture} & Std Acc & Robust Acc ($\ell_\infty$) \\ \hline
Diffpure \citep{nie2022diffusion} & WRN-28-10 & 89.20 & 78.53 \\ \hline
\rowcolor{greyL} AMRM-Pure$_{\text{MAE}}$       & WRN-28-10 & 88.57 & 78.89 \\
\rowcolor{greyL} AMRM-Pure$_{\text{MaskDiT}}$   & WRN-28-10 & 92.03 & 83.44 \\
\rowcolor{greyL} RAMRM-Pure$_{\text{MAE}}$      & WRN-28-10 & 90.09 & 80.17 \\
\rowcolor{greyL} RAMRM-Pure$_{\text{MaskDiT}}$  & WRN-28-10 & \textbf{93.11} & \textbf{85.41} \\ \hline
\end{tabular}
\label{table_bpda}
\end{table}

Table \ref{table_bpda} presents the comparison of clean and robust accuracy under BPDA attack ($\epsilon = \frac{8}{255}$) on the CIFAR-10 dataset. While the conventional Diffpure method achieves decent robustness, our proposed methods demonstrate significant improvements in both clean and robust accuracy. In particular, RAMRM-Pure$_{\text{MaskDiT}}$ achieves the highest clean accuracy (93.11\%) and robust accuracy (85.41\%), highlighting its superior purification capability and enhanced resistance to adversarial attacks. These results validate the effectiveness of our approach in improving semantic reconstruction and adversarial robustness.

\subsection{A specific attacks}

We consider a white-box adversarial attack that jointly optimizes reconstruction fidelity and classification error by minimizing a weighted combination of reconstruction loss and negative classification loss. Specifically, given an input $x$, the attacker generates an adversarial example $x_{\mathrm{adv}}$ by solving the following optimization problem:
\begin{equation}
x_{\mathrm{adv}} = \arg\min_{x'} \; (1 - \alpha)\,\mathcal{L}_{\mathrm{rec}}(x') \;-\; \alpha\,\mathcal{L}_{\mathrm{cls}}(x')
\end{equation}
where $\mathcal{L}_{\mathrm{rec}}$ denotes the reconstruction loss and $\mathcal{L}_{\mathrm{cls}}$ is the classification loss that encourages misclassification. The trade-off parameter $\alpha \in [0,1]$ controls the balance between preserving reconstruction quality and inducing misclassification. The attack is implemented using a strong PGD-based procedure with sufficient iterations and step size tuning to avoid gradient obfuscation. This formulation generalizes standard adversarial attacks by incorporating a reconstruction constraint, aiming to produce adversarial examples that remain visually consistent while still fooling the classifier.

\paragraph{Results and Analysis.}
We evaluate the effectiveness of the proposed joint-loss attack under different settings of the trade-off coefficient $\alpha$ and step size $\beta$.

\textbf{Effect of $\alpha$.}
We vary $\alpha \in \{0.1, 0.3, 0.5, 0.7, 0.9\}$ to study the trade-off between reconstruction fidelity and adversarial strength. The results are shown in Table~\ref{tab:alpha}.

\begin{table}[h]
\caption{Effect of the trade-off parameter $\alpha$ on the joint-loss attack.}
\centering
\small
\begin{tabular}{c|ccccc}
\toprule
$\alpha$ & 0.1 & 0.3 & 0.5 & 0.7 & 0.9 \\
\midrule
AMRM-Pure$_{\text{MaskDiT}}$ & 89.77 & 88.29 & 87.54 & 87.10 & 87.24 \\
RAMRM-Pure$_{\text{MaskDiT}}$ & 90.17 & 89.25 & 89.57 & 89.11 & 88.27 \\
\bottomrule
\end{tabular}
\label{tab:alpha}
\end{table}

As $\alpha$ increases, the adversarial objective becomes stronger, but the reconstruction quality degrades. Specifically, when $\alpha = 0.1$, the reconstruction loss is low ($\mathcal{L}_{rec}^{adv} = 0.04$), but the negative classification loss remains small ($-\mathcal{L}_{cls} = -0.19$), indicating a weak attack. In contrast, when $\alpha = 0.9$, the attack becomes much stronger ($-\mathcal{L}_{cls} = 4.4$), but the reconstruction loss increases significantly ($\mathcal{L}_{rec}^{adv} = 0.19$), leading to severely degraded reconstructions. These results reveal a clear trade-off between reconstruction fidelity and adversarial effectiveness.

\textbf{Effect of $\beta$.}
We further vary the step size $\beta \in \{0.01, 0.05, 0.1, 0.5, 1\}$, and report the results in Table~\ref{tab:beta}.

\begin{table}[h]
\caption{Effect of the step size $\beta$ on the joint-loss attack.}
\centering
\small
\begin{tabular}{c|ccccc}
\toprule
$\beta$ & 0.01 & 0.05 & 0.1 & 0.5 & 1 \\
\midrule
AMRM-Pure$_{\text{MaskDiT}}$ & 84.49 & 83.22 & 87.54 & 88.59 & 89.27 \\
RAMRM-Pure$_{\text{MaskDiT}}$ & 85.59 & 87.92 & 89.57 & 88.62 & 89.54 \\
\bottomrule
\end{tabular}
\label{tab:beta}
\end{table}

In table \ref{tab:beta}, we observe that the attack becomes more effective when $\beta$ is around $0.05$--$0.1$. However, even under the best step size, the proposed attack still does not outperform the standard full-gradient white-box attack used in our paper.

From the above results, we conclude that the joint-loss attack struggles to achieve a satisfactory balance between reconstruction fidelity and adversarial effectiveness. In particular, no choice of $\alpha$ can simultaneously ensure low reconstruction loss and strong adversarial impact. This suggests that the reconstruction objective and classification objective are inherently conflicting under this formulation, making the attack weaker than a standard full-gradient white-box attack.

\subsection{Ablation study}

\subsubsection{Impact of time step number}

To investigate the impact of time steps on the denoising process in MaskDiT, we conduct experiments by observing the robust accuracy at different time steps, aiming to understand how varying time steps influence the model’s ability to effectively remove noise and improve overall performance.

\begin{table}[!h]
\centering
\caption{Impact of time step on CIFAR-10, and all configurations align with Table \ref{cifar_10_full}.}
\begin{tabular}{c|lccl}
\midrule[1.2pt]
Time steps                   & \multicolumn{1}{c}{15} & 20    & 25             & 30             \\ \midrule
AMRM-Pure$_{\text{MaskDiT}}$  & 49.27                  & 50.41 & 50.57          & \textbf{51.69} \\
RAMRM-Pure$_{\text{MaskDiT}}$ & 49.22                  & 51.34 & \textbf{62.13} & 58.22          \\ \midrule[1.2pt]
\end{tabular}
\label{step_num}
\end{table}

As shown in Table \ref{step_num}, our method exhibits a stable increase in robust accuracy as time steps increase, peaking at \textbf{51.69} when the metric of time steps is set as 30. In contrast, RAMRM-Pure$_{\text{MaskDiT}}$ achieves its highest accuracy of \textbf{62.13} when the metric of time steps is set as 25, but experiences a slight drop at the 30$^{th}$ step, indicating its sensitivity to the optimal time step selection.

\subsubsection{Ablation on the Effectiveness of MaskDiT framework}

To explore the origins of the enhanced performance observed with AMRM-Pure$_{\text{MaskDiT}}$ and RAMRM-Pure$_{\text{MaskDiT}}$, we conducted a series of controlled experiments in Table \ref{supplyment_ours}. We begin with training two models, e.g., MaskDiT and RMaskDiT, used for purification. Firstly, following the DiffPure \citep{nie2022diffusion}, we conduct the reconstruction-based purification processes for purification, which are MaskDiT$_{\text{/purification}}$ and RMaskDiT$_{\text{/purification}}$. This reconstruction based purification involves a forward (noise addition) and backward (denoising) pass. The reconstructed images are treated as purified results to investigate whether the superior performance stems from generative ability of selected models. On the other hand, we then compared the above performances with the proposed AMV-based purification methods by using the same models, where the purification processes are named as AMRM-Pure$_{\text{MaskDiT}}$ and RAMRM-Pure$_{\text{MaskDiT}}$. By keeping the model architecture consistent across all variants, we enabled a direct comparison of the effectiveness of different denoising strategies. It is note that MaskDiT and RMaskDiT indeed possess a certain capability to withstand stronger adversarial attacks; however, their effectiveness is still much lower than that of our proposed AMRM-Pure$_{\text{MaskDiT}}$ and RAMRM-Pure$_{\text{MaskDiT}}$.


\begin{table}[h]
\centering
\footnotesize
\caption{Ablation analysis on different denoising components across various datasets.}
\begin{tabular}{c|c|cc|cc|cc}
\midrule[1.2pt]
\multirow{2}{*}{Method}           & \multirow{2}{*}{Architecture} & \multicolumn{2}{c|}{CIFAR10}    & \multicolumn{2}{c|}{CIFAR100}   & \multicolumn{2}{c}{SVHN}        \\ \cline{3-8} 
                                  &                               & Std acc        & Robust acc     & Std acc        & Robust acc     & Std acc        & Robust acc     \\ \midrule
MaskDiT$_{/\text{purification}}$  &WRN-28-10              & 91.13          & 42.99          & 63.29          & 13.57          & 92.28          & 40.07          \\
RMaskDiT$_{/\text{purification}}$ & WRN-28-10              & 90.57          & 47.57          & 64.15          & 18.55          & 93.03          & 45.19          \\
\cellcolor{greyL}AMRM-Pure$_{\text{MaskDiT}}$                      & \cellcolor{greyL}WRN-28-10              & \cellcolor{greyL}92.03          & \cellcolor{greyL}50.57          & \cellcolor{greyL}\textbf{70.03} & \cellcolor{greyL}24.39          & \cellcolor{greyL}94.91          & \cellcolor{greyL}46.57          \\
\cellcolor{greyL}RAMRM-Pure$_{\text{MaskDiT}}$                     & \cellcolor{greyL}WRN-28-10              & \cellcolor{greyL}\textbf{93.11} & \cellcolor{greyL}\textbf{62.13} & \cellcolor{greyL}69.87          & \cellcolor{greyL}\textbf{29.91} & \cellcolor{greyL}\textbf{95.39} & \cellcolor{greyL}\textbf{55.90} \\ \midrule[1.2pt]
\end{tabular}
\label{supplyment_ours}
\end{table}

\subsection{Sensitivity Analysis}
\label{senstiv}
In this subsection, taking the CIFAR10 under same setting with Table \ref{cifar_10_full}, we analyze the impact of step size, mask ratio and step size. The result is described in Fig. \ref{sens_STEP}. It confirms our theoretical analysis.
\begin{figure}[H]
  \begin{minipage}[b]{0.23\linewidth}
    \centering
    \subfloat[Step size\label{or1aa}]{\includegraphics[width=\linewidth]{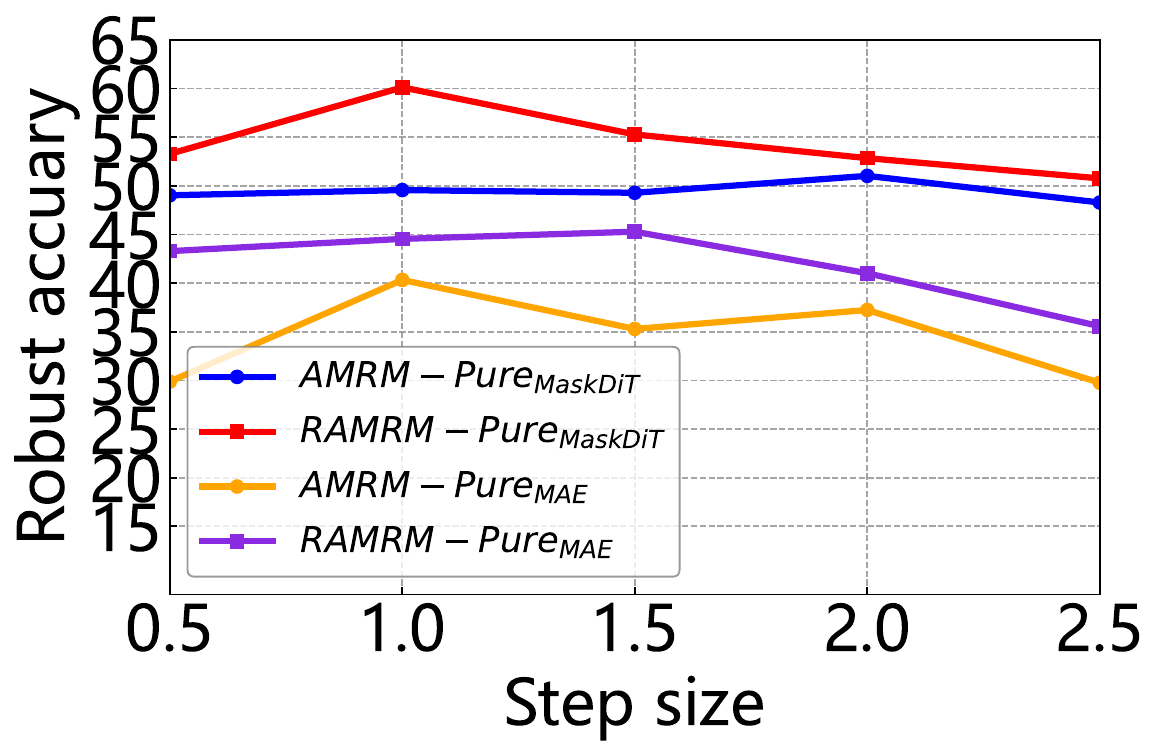}}
  \end{minipage}
   \hfill
  \begin{minipage}[b]{0.23\linewidth}
    \centering
    \subfloat[Mask ratio\label{m1aa}]{\includegraphics[width=\linewidth]{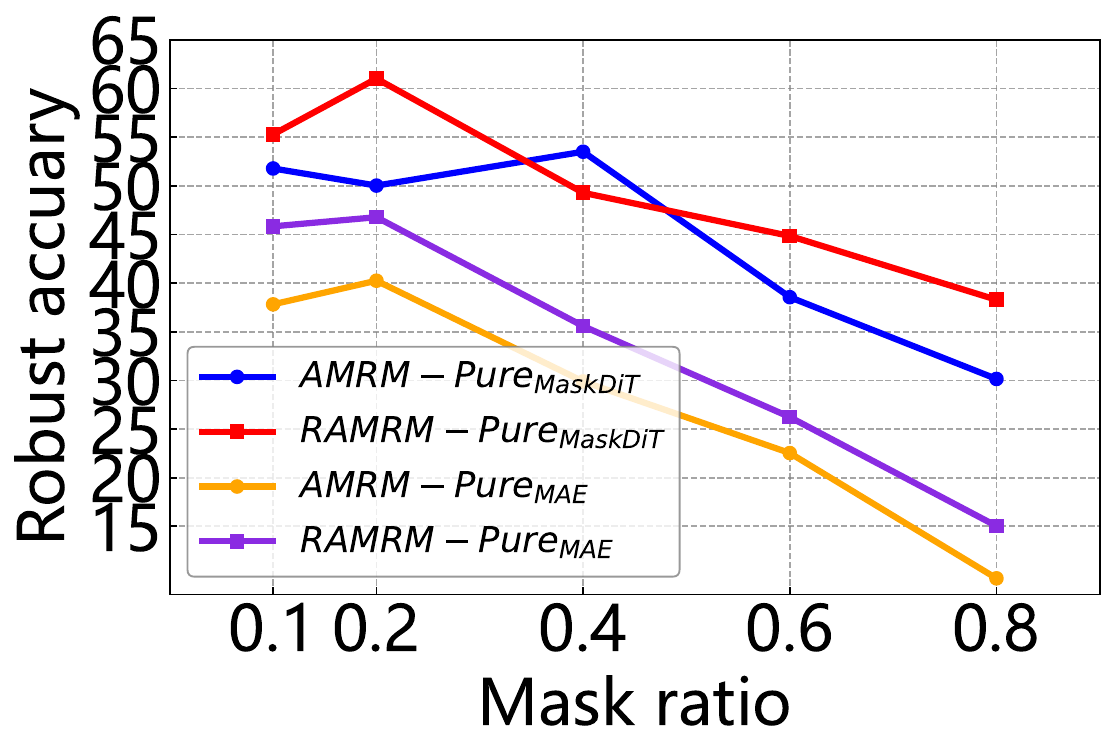}}
  \end{minipage}
   \hfill
   \begin{minipage}[b]{0.23\linewidth}
    \centering
    \subfloat[Iteration /w MAE\label{m1aa}]{\includegraphics[width=\linewidth]{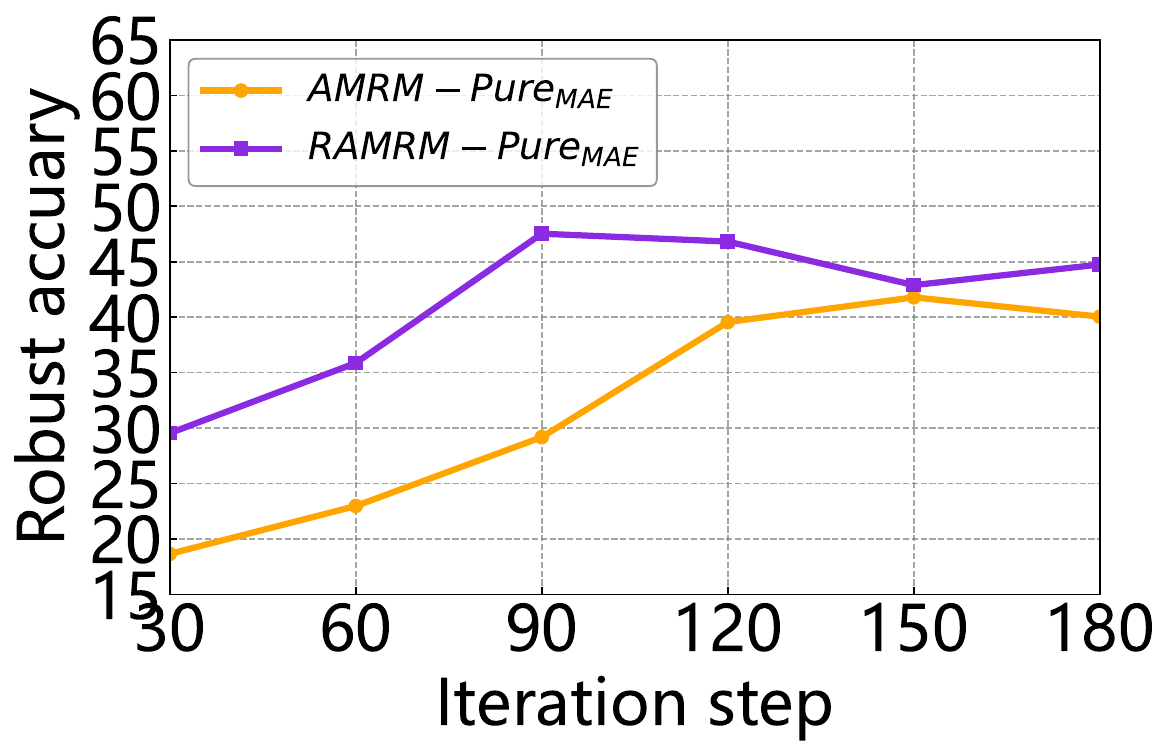}}
  \end{minipage}
   \hfill
   \begin{minipage}[b]{0.23\linewidth}
    \centering
    \subfloat[Iteration /w MaskDiT \label{m1aa}]{\includegraphics[width=\linewidth]{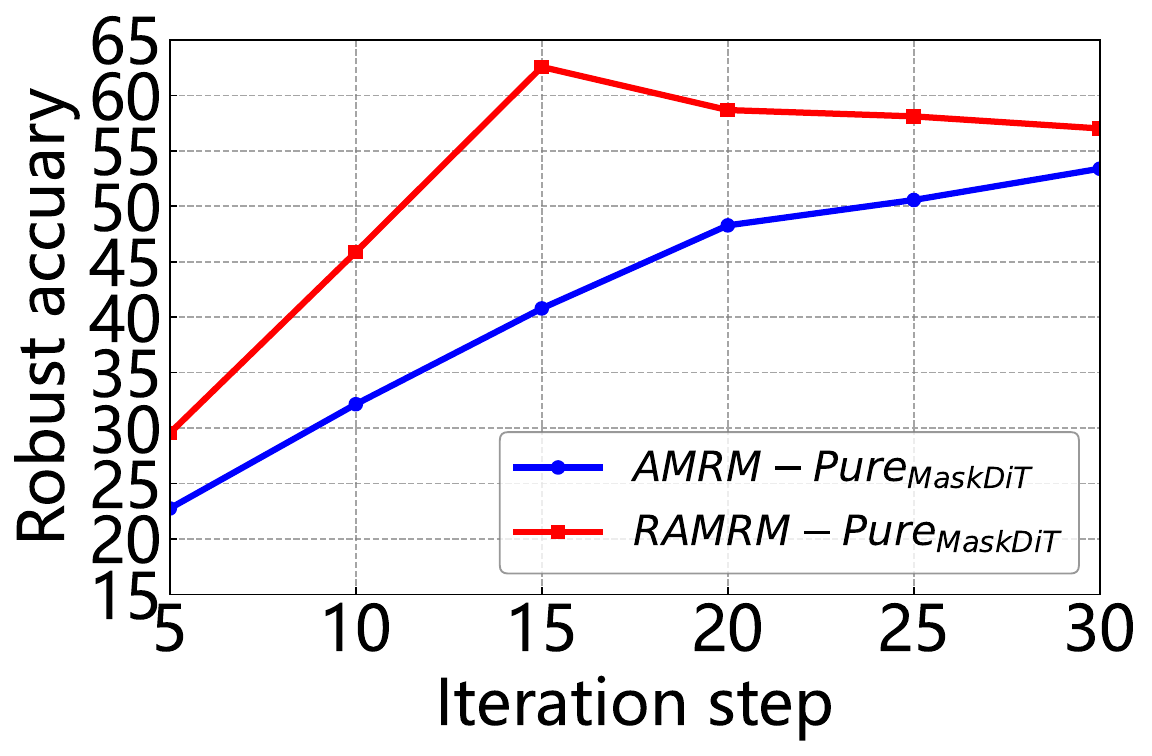}}
  \end{minipage}
   \hfill
  \caption{Sensitivity analysis}
  \label{sens_STEP}
\end{figure}

\subsubsection{Is non-attentive MRM effective?}

To investigate whether non-attentive MRMs are effective, we adopt a CNN-based architecture similar to MAE, namely Context Encoders (CE) \citep{pathak2016context}, and propose CE-Pure in a manner analogous to AMRM-Pure$_{\text{MAE}}$. We show their results in Table \ref{cnn_based} Context Encoders, like MAE, are self-supervised methods based on masking and reconstruction, but differ in architectural design: CE relies on CNNs that excel at capturing local structural information and are widely used for image inpainting, while MAE is Transformer-based and better at modeling global contextual information. We evaluate CE-Pure on CIFAR-10 and SVHN using WRN-28 as the classifier, and adopt PGD200+EOT20 with full gradient for consistency with the main paper. The results are as follows:

\begin{table}[ht]
\centering
\caption{Comparison of CE-Pure and AMRM-Pure$_{\text{MAE}}$ results on CIFAR-10 and SVHN using WRN-28 as the classifier with Sta Acc and Robust Acc ($\ell_{\infty}=\frac{8}{255}$).}
\begin{tabular}{c|cc|cc}
\hline
\multirow{2}{*}{Method} & \multicolumn{2}{c|}{CIFAR-10}   & \multicolumn{2}{c}{SVHN}        \\ \cline{2-5} 
                        & Std Acc        & Robust Acc     & Std Acc        & Robust Acc     \\ \hline
CE-Pure                 & 82.59          & 6.75           & 85.33          & 8.59           \\
\cellcolor{greyL}AMRM-Pure$_{\text{MAE}}$                & \cellcolor{greyL} \textbf{88.57} & \cellcolor{greyL} \textbf{40.53} & \cellcolor{greyL} \textbf{95.39} & \cellcolor{greyL} \textbf{55.90} \\ \hline
\end{tabular}
\label{cnn_based}
\end{table}

The results show that the robustness of AMRM-Pure$_{\text{MAE}}$ is significantly stronger than that of CE-Pure. This mechanism cannot be transferred to CNNs, thus proving that AMRM has better robustness.


\section{Robust MAE and MaskDiT analysis}
\label{robust:purification}

\begin{figure}[h]
  \centering
  \subfloat[ \label{fig:purification_iterations}]{
    \includegraphics[width=0.45\linewidth]{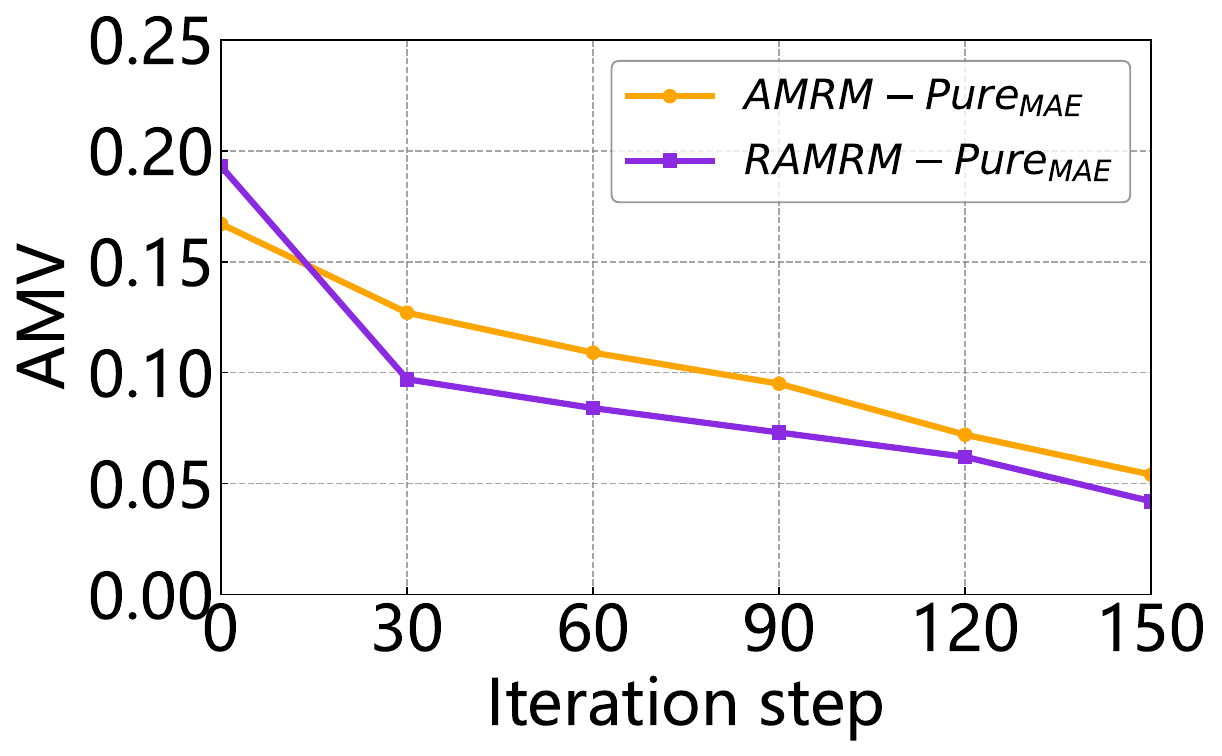}
  }
  \hfill
  \subfloat[ \label{fig:purification_iterations_md}]{
    \includegraphics[width=0.45\linewidth]{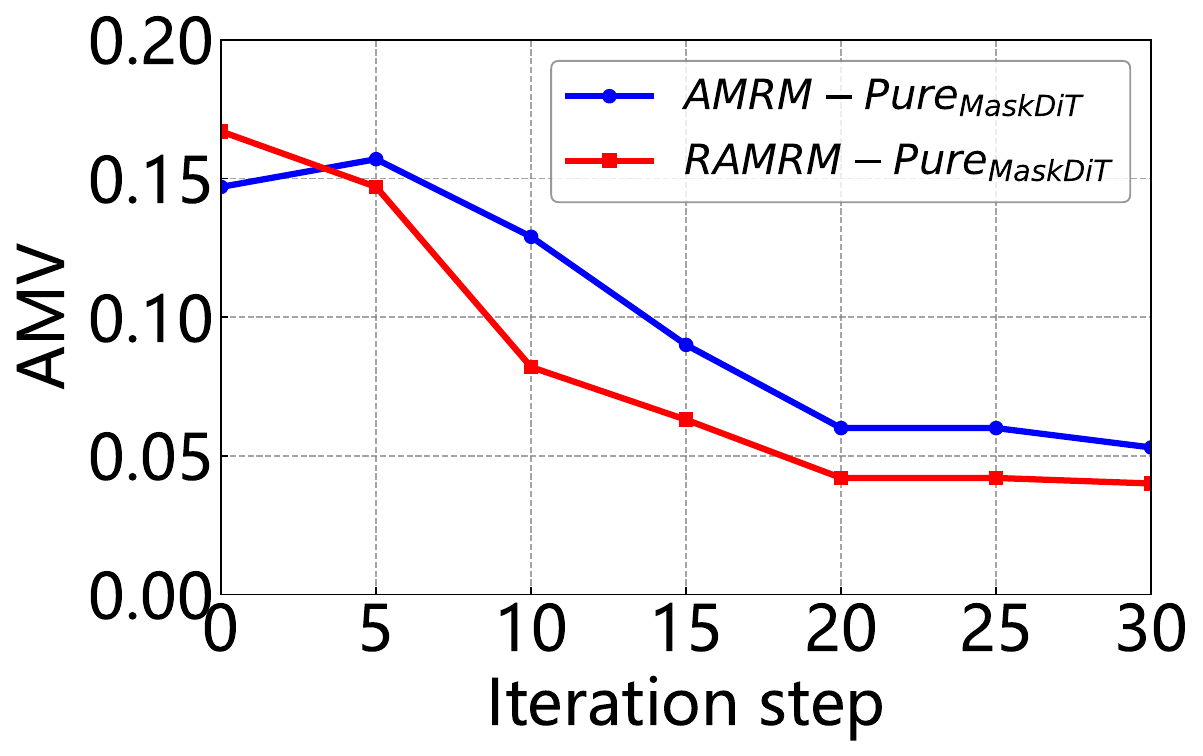}
  }
  \caption{Impact of purification iterations on the AMV using CIFAR-10.}
  \label{fig:combined_purification}
\end{figure}

\subsection{Fine-tuning of RMRM-Pure}
\label{robust:purification}

\textbf{Stage I:} We begin by generating adversarial example dataset $\text{X}_{adv}'$, and it achieves the purifier-classifier system as follows:
\begin{equation}
    \text{X}_{adv}' =  \max_{\delta}\left[\sum_{(x,y)\in \text{X}} \mathcal{L}_\mathcal{C}\left(\mathcal{P}_{\theta}(x'+\delta),y\right)\right],
    \label{miniz}
\end{equation}
where \( x' \) represents the perturbed sample \( x \) with added Gaussian noise, and \( y \) denotes its corresponding label from the training dataset \( \text{X} \). 
The functions \( \mathcal{P}_{\theta}(\cdot) \) and \( \mathcal{L}_\mathcal{C}(\cdot) \) correspond to the purification process and the loss of the classifier, respectively.
$x'_{adv}$ is an adversarial sample designed to target the entire purifier-classifier system. It is utilized during the subsequent fine-tuning stage to improve the overall robustness of the system. 

\textbf{Stage II:} We choose to use the generated adversarial example $x'_{adv}$ from adversarial dataset $\text{X}_{adv}'$ for fine-tuning the purification model as follows:
\begin{equation}
\min_{\theta} \mathcal{L}_{fine}(x_{adv}',y,\theta) = \min_{\theta}\sum_{(x_{adv}',y)\in \text{X}_{adv}'}\mathcal{L}_\mathcal{C}\left(x_{adv}',y\right),
\label{fine_tune}
\end{equation}
where $\theta$ represents the weight of the purification model. 

Compared to AMRM-Pure, RAMRM-Pure optimizes Eq. (\ref{miniz}) to steer denoised images toward the classifier domain of natural datasets. 
This operation ensures that the attention distribution of denoised images more closely resembles that of natural samples, thereby reducing AMV and preserving inter-patch semantic relationships. As a result, RAMRM-Pure achieves superior robustness and generalization.

\section{Validation of MaskDiT}
\label{maskdit_val}

\begin{figure*}[h]
  \centering
  \subfloat[ \label{fig:maskdit_per}]{
    \includegraphics[width=0.45\linewidth]{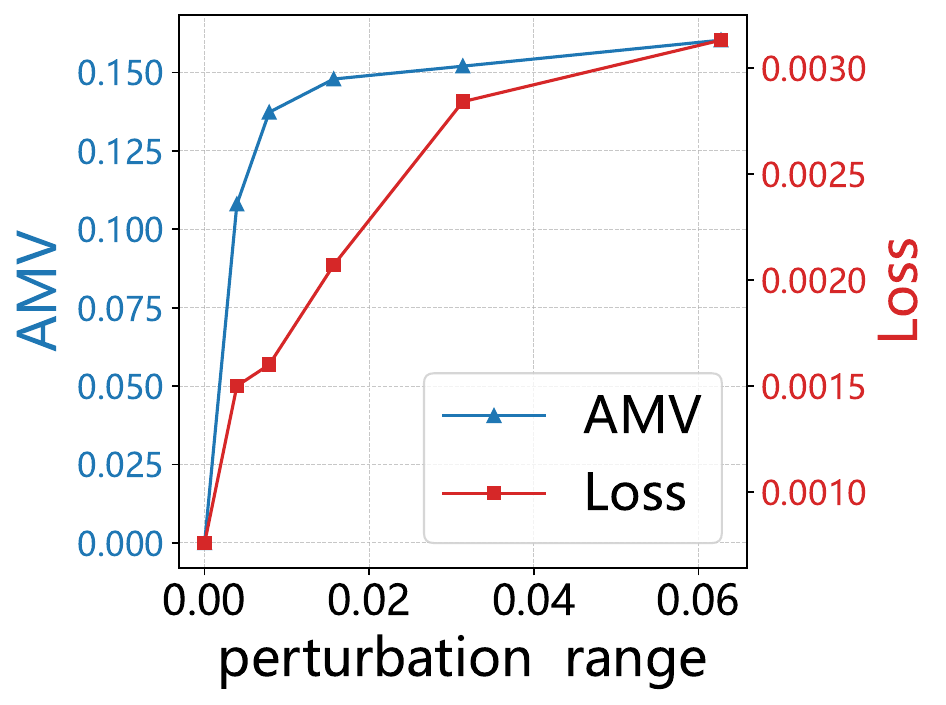}
  }
  \hfill
  \subfloat[ \label{fig:maskdit_pur}]{
    \includegraphics[width=0.45\linewidth]{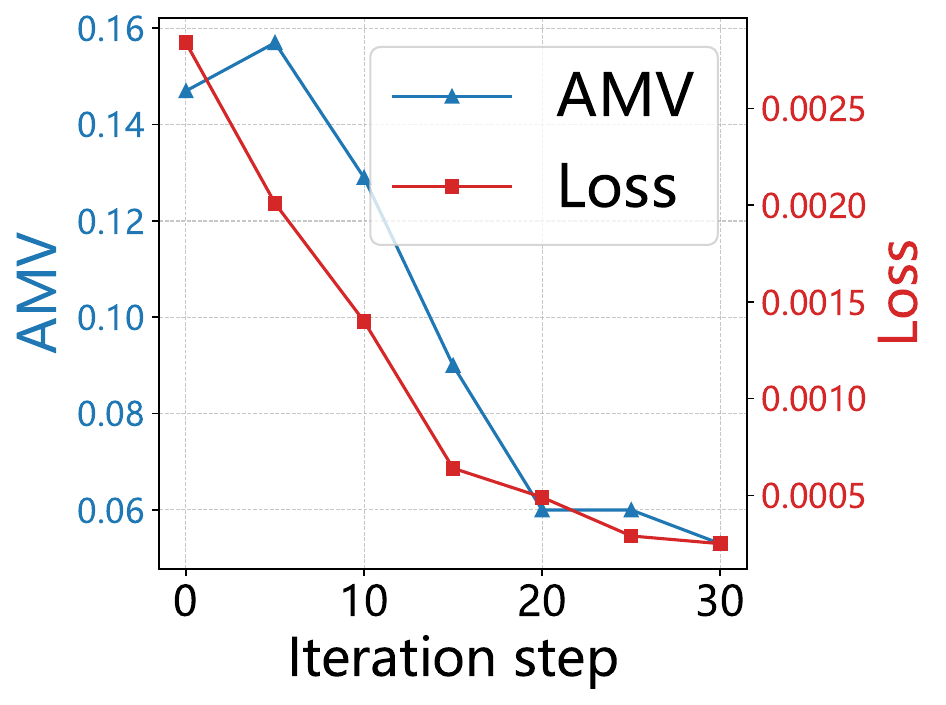}
  }
  \caption{The relationship between the trends of MaskDit loss, AMV, and perturbation.}
  \label{fig:combined_purification_maskdit}
\end{figure*}

Figure \ref{fig:combined_purification_maskdit} compares the reconstruction loss, AMV, and perturbation of MaskDiT with those of MAE shown in Figures \ref{trens_forard}, revealing similar patterns. Like MAE, MaskDiT is highly sensitive to noise, with AMV increasing sharply under minimal perturbations (e.g., $\delta = \frac{1}{255}$). Based on this observation, we extend AMRM-Pure$_{\text{MAE}}$ manner to MaskDiT.

\section{Algorithm}
\label{sec:algo}

\par Algorithm \ref{attack_algo} describes the AMRM-Pure defense procedure. The fine-tuning process of RAMRM-Pure consists of two stages.
\begin{algorithm}[H]
\caption{AMRM-Pure.}
    \label{attack_algo}
    \small
\begin{algorithmic}[1]
    \renewcommand{\algorithmicrequire}{\textbf{Input:}}
    \renewcommand{\algorithmicensure}{\textbf{Output:}}
   \REQUIRE Adversarial Example $x_{adv}$,  Step Size $\lambda$, Number of iteration $S$, clipping threshold $\eta$.
    \ENSURE Denoised data $x_{den}$.
    
     \STATE $\text{s} \gets 0$ \\
     \STATE  {$ x_{adv}^{s} \gets x_{adv}$}\\
     \WHILE  {$\text{s} \le \text{S}$}
     \STATE $\text{s} \gets \text{s} +1$ \\
     \STATE Gain the AMRM reconstruction loss for adversarial examples $\mathcal{L}_{\text{rec}}(x_{adv}^{s})$
     \STATE $ \Delta_{s} = \text{sign} (\nabla_{x}  \mathcal{L}_{\text{rec}}(x_{adv}^{s})) $
      \STATE $ x_{adv}^{s} \gets clip(x_{adv}^{s} - \lambda \Delta_s,\eta)$
    \ENDWHILE \\
    \STATE $x_{den} \gets x_{adv}^{S}$
\end{algorithmic} 
\label{alg1}
\end{algorithm}

\section{Visualization Analysis}
\label{visualization_more}

\par The visualization analysis of our proposed method is illustrated as Fig. \ref{amv_}. It qualitatively verifies the effectiveness of our proposed method.
\begin{figure}[H]

    
  \begin{minipage}[b]{0.22\linewidth}
    \centering
    \subfloat[\label{m_2}]{\includegraphics[width=\linewidth]{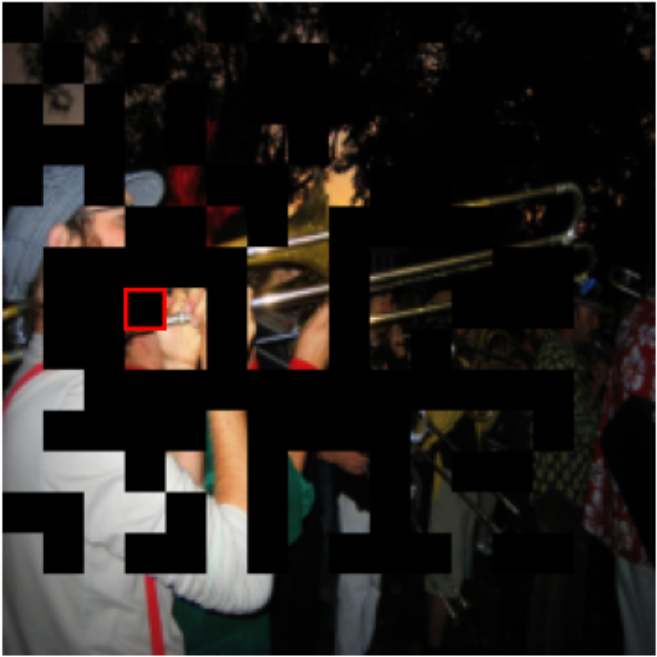}}
  \end{minipage}
   \hfill
  \begin{minipage}[b]{0.22\linewidth}
    \centering
    \subfloat[\label{o_2}]{\includegraphics[width=\linewidth]{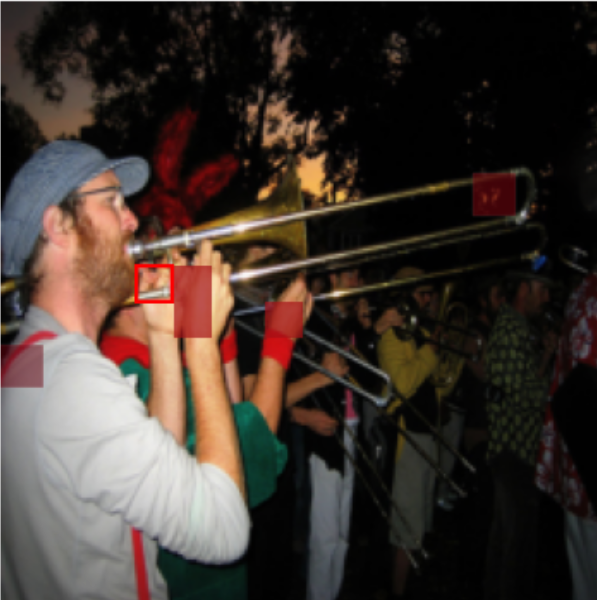}}
  \end{minipage}
   \hfill
  \begin{minipage}[b]{0.22\linewidth}
    \centering
    \subfloat[\label{adv_2}]{\includegraphics[width=\linewidth]{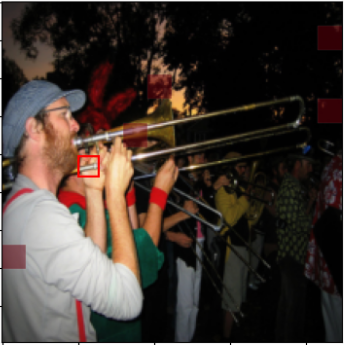}}
  \end{minipage}
  \hfill
  \begin{minipage}[b]{0.22\linewidth}
    \centering
    \subfloat[\label{purification_2}]{\includegraphics[width=\linewidth]{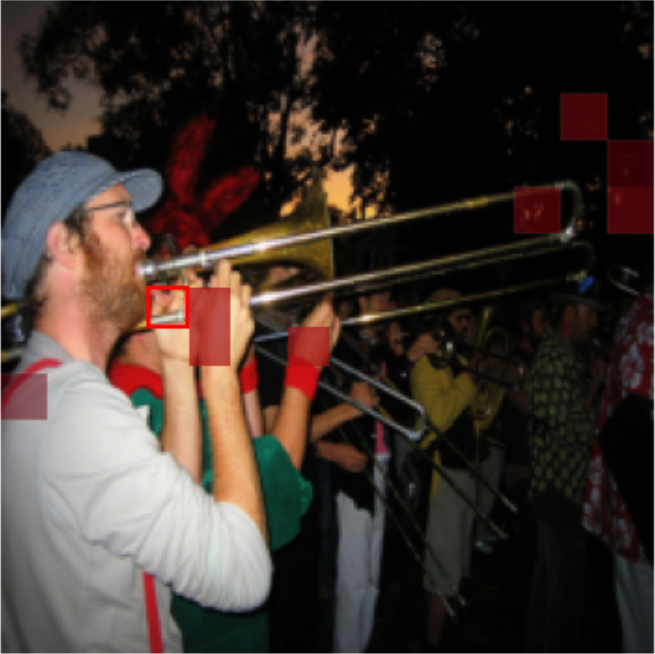}}
  \end{minipage}
  \hfill
   \begin{minipage}[b]{0.045\linewidth}
    \centering
    \vspace*{0.1cm}
    {\includegraphics[width=\linewidth]{figs/2d.jpg}}
  \end{minipage}
  \hfill

  \begin{minipage}[b]{0.22\linewidth}
    \centering
    \subfloat[\label{m_3}]{\includegraphics[width=\linewidth]{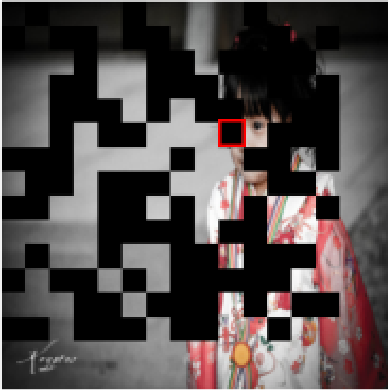}}
  \end{minipage}
   \hfill
  \begin{minipage}[b]{0.22\linewidth}
    \centering
    \subfloat[\label{o_3}]{\includegraphics[width=\linewidth]{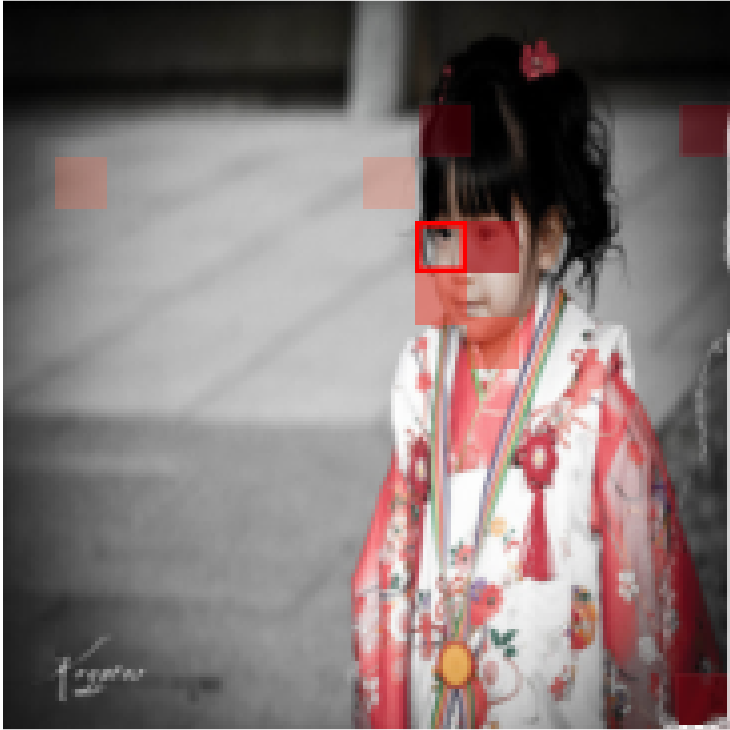}}
  \end{minipage}
   \hfill
  \begin{minipage}[b]{0.22\linewidth}
    \centering
    \subfloat[\label{adv_3}]{\includegraphics[width=\linewidth]{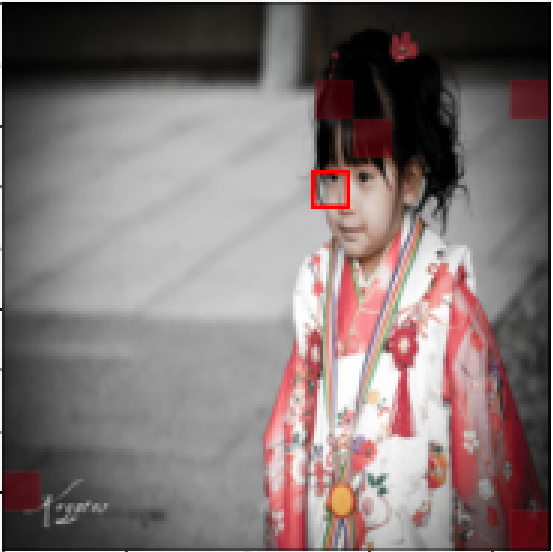}}
  \end{minipage}
    \hfill
  \begin{minipage}[b]{0.22\linewidth}
    \centering
    \subfloat[\label{purification_3}]{\includegraphics[width=\linewidth]{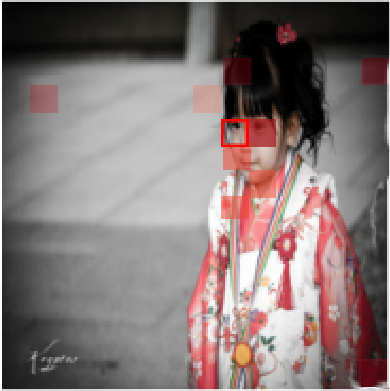}}
  \end{minipage}
  \hfill
   \begin{minipage}[b]{0.045\linewidth}
    \centering
    {\includegraphics[width=\linewidth]{figs/2d.jpg}}
  \end{minipage}
    \hfill

  \caption{
  The first column, Fig. (a) and (e), represents the Mask Matrix.
The second column, Fig. (b) and (f), illustrates the Attention Weights for clean samples.
The third column, Fig. (c) and (k), depicts the Attention Weights for adversarial examples.
The fourth column, Fig. (d) and (l), showcases the Attention Weights for denoised samples (by our AMRM-Pure$_{\text{MAE}}$). 
  Patches with a deeper red color mean the elements with more attention. The data is sampled from the ImageNet dataset~\citep{deng2009imagenet}.}
  \label{amv_}
    \vspace{-.15in}
\end{figure}

\section{Convergence Analysis of Purification Process.} 
\label{Convergence}

For a well-trained MAE model, the reconstruction loss $\mathcal{L}_{\text{rec}}(x)$ is expected to reach a local minimum $\mathcal{L}_{\text{rec}}^*$ ,where the input $x$ corresponds to a clean example, i.e., $\mathcal{L}_{\text{rec}}(x) \approx \mathcal{L}_{\text{rec}}^*$.
Motivated by prior analysis ~\citep{lee2024visualizing, zhang2022mask} that the loss landscape of MAE is smoother and exhibits wider convex regions, the reconstruction loss $\mathcal{L}_{\text{rec}}$ can be regarded as weakly convex within the neighborhood $[x - \delta, x + \delta]$ around the clean input $x$.
Purified samples at the $s$-th iteration of the denoising process are represented as $x_{adv}^s$, and their corresponding reconstruction loss is denoted as $\mathcal{L}_{\text{rec}}^s=\mathcal{L}_{\text{rec}}(x_{adv}^s, \mathbf{U})$, where $\mathbf{U}\in \mathbb{R}^{n\times n}$ represents the MAE mask.

\begin{theorem}	
\label{thm3}
      Let $\{\mathbf{U}_i\}_{i=1}^{E}$ be mask set which contains all possible masks with mask ratio $\rho$. After $S$ optimization iterations  according to Eq.\ref{results},  it holds that:
      {\scriptsize 
      \begin{align*}
    & \frac{1}{(1-\rho)E}\sum_{e=1}^{E}[\mathcal{L}_{\text{rec}}(\frac{1}{S+1}\sum_{s=0}^{S}x_{adv}^s, \mathbf{U}_e) - \mathcal{L}_{\text{rec}}^*] 
   \leq  \frac{1}{(1-\rho)E} \sum_{e=1}^E[\frac{||x_{adv}-x||^2_2}{2\lambda (S+1)} + \frac{\lambda}{2( S+1)}\sum_{s=0}^{S}||\nabla_x\mathcal{L}_{\text{rec}}(x_{adv}^{s}, \mathbf{U}_e)||^2_2], \nonumber
\end{align*}
} 
where $\lambda$ is the step size and $\eta$ is the clipping threshold. 
\end{theorem}

\begin{proof}
The proof is relegated to Supplementary Material I.4.
\end{proof}


Theorem~\ref{thm3} provides an upper bound on the gap between the averaged reconstruction loss during purification and the optimal loss $\mathcal{L}_{\text{rec}}^*$. As $S$ increases, the first term vanishes at $\mathcal{O}(1/S)$ and the second term decreases as the gradient norm $\|\nabla_x \mathcal{L}_{\text{rec}}(x_{adv}^{s}, \mathbf{U}_e)\|_2^2$ diminishes. Hence, the RHS tends to zero, implying that the LHS also converges to zero. As a result, $\mathcal{L}_{\text{rec}}\left(\frac{1}{S+1} \sum_{s=0}^{S} x_{adv}^s, \mathbf{U}_e\right) \to \mathcal{L}^*_{\text{rec}}(x)$, indicating that the denoised sample progressively approximates the clean example in the reconstruction space.

\section{Proof of Theoretical Analysis}

\subsection{Assumption}
\label{exp:assumption}

\noindent \textbf{Assumption 1.} \textit{There exists a pseudo-inverse encoder \(f_g\) that satisfies \(\|g(f_g(a)) - a\|_2 \leq c_{rec}\) for any non-degenerate decoder, where $a$ can be ether $x_1$ or $x_2$.} Here,  $x_1$ and $x_2$ indicate visible portion and masked patches of input image $x$. $c_{rec}$ is reconstruction bias, which symbolizes the disparity between the output of MAE and the original, unmasked image.

\noindent \textbf{Remark.}  
To facilitate the theoretical analysis of MAE, we borrow the above reasonable assumption from the previous study~\citep{zhang2022mask}. Intuitively, this assumption states that within the MAE framework, the decoder \(g\) is non-degenerate---i.e., its outputs retain meaningful information---and there exists a corresponding pseudo-inverse encoder \(f_g\), such that their composition \(h_g = g \circ f_g\) can approximately reconstruct either the visible patches \(x_1\) or the masked patches \(x_2\) of the input image. Physically, this implies that the decoder in MAE has sufficient capacity for recovering local structures from latent representations, a property empirically verified in many Transformer-based autoencoding models. This assumption provides the theoretical foundation for connecting MAE’s reconstruction loss with alignment loss and helps interpret MAE as implicitly performing contrastive alignment through its masking mechanism.

\noindent \textbf{Assumption 2 (Lipschitz Continuity).}
Let $\mathbf{A}^{\mathrm{dec}}_{i} = \Phi\bigl(\{\mathbf{A}^{\mathrm{dec}}_{i,t}\}_{t=1}^T\bigr)$ denote the final decoder attention matrix obtained from the per-layer attention matrices $\{\mathbf{A}^{\mathrm{dec}}_{i,t}\}_{t=1}^T$. We assume $\Phi$ is $L$-Lipschitz in the domain of interest, i.e., there exists a constant $L>0$ such that for any sets of matrices $\{\mathbf{X}_t\}_{t=1}^T$ and $\{\mathbf{Y}_t\}_{t=1}^T$:
\[
\bigl\|\Phi(\mathbf{X}_1,\dots,\mathbf{X}_T)
-\Phi(\mathbf{Y}_1,\dots,\mathbf{Y}_T)\bigr\|
\;\le\;
L\sum_{t=1}^T\|\mathbf{X}_t-\mathbf{Y}_t\|.
\]
Under this assumption, for the adversarially perturbed attention matrices $\{\mathbf{A}^{\mathrm{dec}}_{\mathrm{adv}, i,t}\}$, the final attention matrix satisfies:
\[
\bigl\|\mathbf{A}^{\mathrm{dec}}_{\mathrm{adv}, i}
-\mathbf{A}^{\mathrm{dec}}_{i}\bigr\|^2
\;\le\;
L^2\,T
\sum_{t=1}^T
\Bigl\|\mathbf{A}^{\mathrm{dec}}_{\mathrm{adv}, i,t}
-\mathbf{A}^{\mathrm{dec}}_{i,t}\Bigr\|^2.
\]
Hence, we can further write as:
\[
\bigl\|\mathbf{A}^{\mathrm{dec}}_{\mathrm{adv}, i}
-\mathbf{A}^{\mathrm{dec}}_{i}\bigr\|^2
\;\;\le\;\;
\frac{H}{T}
\sum_{t=1}^T
\Bigl\|\mathbf{A}^{\mathrm{dec}}_{\mathrm{adv}, i,t}
-\mathbf{A}^{\mathrm{dec}}_{i,t}\Bigr\|^2,
\]
by setting $H = L^2 T^2$, thereby bounding the overall adversarial effect via the per-layer differences.

\subsection{Proof of Theorem 3.1}
\label{proof_thm1}


\textit{Proof:}

\begin{align*}
||\mathbf{A}^{t}_{\text{adv}} - \mathbf{A}^{t}||_2 
&= \left\| \text{softmax}(\mathbf{Q}^{t}_{\text{adv}}(\mathbf{K}^{t}_{\text{adv}})^T) - \text{softmax}(\mathbf{Q}^{t}(\mathbf{K}^{t})^T) \right\|_2
\end{align*}

To streamline the equation, we employ kernel methods as a substitute for the $\text{softmax}(\cdot)$ function~\citep{choromanski2020rethinking}. Let the kernel function be denoted as $\phi(\cdot)$:

\begin{align*}
\text{softmax}(\mathbf{QK}^{T}) \approx \langle \phi(\mathbf{Q}), \phi(\mathbf{K}) \rangle
\end{align*}

The definition is written as follows:
\[
\phi(x) = \frac{d(x)}{\sqrt{k}}\{f(\omega_1^\top x),...,f(\omega_k^\top x)\}, \quad
d(x) = \exp\left(-\frac{||x||^2}{2}\right), \quad f(x) = \exp(x)
\]

Thus, we obtain:
\begin{align*}
\mathbf{A}^{t}_{\text{adv}} - \mathbf{A}^{t} 
&\approx \frac{1}{m} \left[
\exp\left(-\frac{||\mathbf{Q}_{\text{adv}}^t||^2+||\mathbf{K}_{\text{adv}}^t||^2}{2}\right) \sum^{k}_{i=0} \exp\left(\omega_i^\top(\mathbf{Q}^{t}_{\text{adv}}+\mathbf{K}^{t}_{\text{adv}})\right) \right. \\
&\quad\left. - \exp\left(-\frac{||\mathbf{Q}^t||^2+||\mathbf{K}^t||^2}{2}\right) \sum^{k}_{i=0} \exp\left(\omega_i^\top(\mathbf{Q}^{t}+\mathbf{K}^{t})\right) \right]
\end{align*}

For: \(\mathbf{Q}_{\text{adv}}^t = \mathbf{Q}^t + \Delta \mathbf{Q}^t\), \(\mathbf{K}_{\text{adv}}^t = \mathbf{K}^t + \Delta \mathbf{K}^t\)

Following \citep{nguyen2022fourierformer,moosavi2016deepfool}, using Taylor expansion approximations, the formulation can be deduced as follows:

\[
\exp\left( -\frac{\|\mathbf{Q}_{\text{adv}}^t\|^2 + \|\mathbf{K}_{\text{adv}}^t\|^2}{2} \right) 
\ge \exp\left( -\frac{\|\mathbf{Q}^t\|^2 + \|\mathbf{K}^t\|^2}{2} \right) 
\left( 1 - (\mathbf{Q}^t)^\top \Delta \mathbf{Q}^t - (\mathbf{K}^t)^\top \Delta \mathbf{K}^t \right)
\]

Let \(\mathbf{S} = \mathbf{Q}^t + \mathbf{K}^t\), \(\Delta \mathbf{S} = \Delta \mathbf{Q}^t + \Delta \mathbf{K}^t\), there exists a following equation:

\[
\sum_{i=0}^{k} \exp\left( \omega_i^\top (\mathbf{S} + \Delta \mathbf{S}) \right)
\ge \sum_{i=0}^{k} \exp\left( \omega_i^\top \mathbf{S} \right) 
+ \sum_{i=0}^{k} \exp\left( \omega_i^\top \mathbf{S} \right) \omega_i^\top \Delta \mathbf{S}
\]

The definition is written as follows:

\[
\mathbf{B} = \sum_{i=0}^{k} \exp\left( \omega_i^\top \mathbf{S} \right), \quad
\mathbf{Y} = \sum_{i=0}^{k} \exp\left( \omega_i^\top \mathbf{S} \right) \omega_i
\]

Then, we deduce the following sum equation:

\[
\sum_{i=0}^{k} \exp\left( \omega_i^\top (\mathbf{Q}_{\text{adv}}^t + \mathbf{K}_{\text{adv}}^t) \right) 
\approx \mathbf{B} + \mathbf{Y}^\top \Delta \mathbf{S}
= \mathbf{B} + \mathbf{Y}^\top (\Delta \mathbf{Q}^t + \Delta \mathbf{K}^t)
\]

The substitution into difference is:

\begin{align*}
D &\approx \frac{1}{m} \exp\left( -\frac{\|\mathbf{Q}^t\|^2 + \|\mathbf{K}^t\|^2}{2} \right) 
\left[ \left( 1 - (\mathbf{Q}^t)^\top \Delta \mathbf{Q}^t - (\mathbf{K}^t)^\top \Delta \mathbf{K}^t \right)
\left( \mathbf{B} + \mathbf{Y}^\top (\Delta \mathbf{Q}^t + \Delta \mathbf{K}^t) \right) - \mathbf{B} \right] \\
&\approx \frac{1}{m} \exp\left( -\frac{\|\mathbf{Q}^t\|^2 + \|\mathbf{K}^t\|^2}{2} \right) 
\left[ \mathbf{Y}^\top (\Delta \mathbf{Q}^t + \Delta \mathbf{K}^t) 
- \mathbf{B} \left( (\mathbf{Q}^t)^\top \Delta \mathbf{Q}^t + (\mathbf{K}^t)^\top \Delta \mathbf{K}^t \right) \right]
\end{align*}

We group and finalize to get a concluding formulation:

\begin{align*}
||\mathbf{A}^t_{\text{adv}} - \mathbf{A}^t||_2 
&\approx \frac{\exp\left( -\frac{\|\mathbf{Q}^t\|^2 + \|\mathbf{K}^t\|^2}{2} \right)}{m} 
\left\| \left( \mathbf{Y} - \mathbf{B} \mathbf{Q}^t \right)^\top \Delta \mathbf{Q}^t 
+ \left( \mathbf{Y} - \mathbf{B} \mathbf{K}^t \right)^\top \Delta \mathbf{K}^t \right\|_2 \\
&= \gamma \left\| \left( \mathbf{Y} - \mathbf{B} \mathbf{Q}^t \right)^\top \Delta \mathbf{Q}^t 
+ \left( \mathbf{Y} - \mathbf{B} \mathbf{K}^t \right)^\top \Delta \mathbf{K}^t \right\|_2
\end{align*}

where the variables' definitions are:
\begin{itemize}
\item \( \mathbf{B} = \sum_{i=0}^{k} \exp\left( \omega_i^\top (\mathbf{Q}^t + \mathbf{K}^t) \right) \),
\item \( \mathbf{Y} = \sum_{i=0}^{k} \exp\left( \omega_i^\top (\mathbf{Q}^t + \mathbf{K}^t) \right) \omega_i \),
\item \( \gamma = \frac{\exp\left( -\frac{\|\mathbf{Q}^t\|^2 + \|\mathbf{K}^t\|^2}{2} \right)}{m} \).\\
\end{itemize}
Therefore, we get 
\begin{equation}
    ||\mathbf{A}_{adv}^{t}-\mathbf{A}^{t}||_2 \ge
    \gamma \left\| \left[\left( \mathbf{Y} - \mathbf{B} \mathbf{Q}^t \right)^\top \mathbf{W}^t_Q + \left( \mathbf{Y} - \mathbf{B} \mathbf{K}^t \right)^\top \mathbf{W}^t_K \right]\delta_t \right\|_2,
\end{equation}

The proof is complete.

\subsection{Proof of Theorem 3.2}
\label{proof:thm2}


\textit{Proof:}

\begin{equation*}
\begin{aligned}
     \mathcal{L}_{\text{rec}}^{adv} &= \frac{1}{N}\sum_{i=1}^{N}||g(f(x_{i_1}^{adv}))-x_{i_2}^{adv} ||^2 \\
\end{aligned}
\end{equation*}

\begin{equation*}
\begin{aligned}
    & =\frac{1}{N}\sum_{i=1}^{N}[||g(f(x_{i_1}^{adv}))-x_{i_2}^{adv}||_2 +c_{\text{rec}} -c_{\text{rec}}] \\
    &\ge \frac{1}{N}\sum_{i=1}^{N}[||g(f(x_{i_1}^{adv}))-x_{i_2}^{adv}||^2 +|| g(f_g(x_{i_2})) -x_{i_2}||^2 -c_{\text{rec}}]  \\
    &= \frac{1}{N}\sum_{1}^{N}[||g(f(x_{i_1}^{adv}))-x_{i_2}^{adv}+g(f(x_{i_1}))-g(f(x_{i_1}))||^2 +|| g(f_g(x_{i_2})) -x_{i_2}||^2 -c_{\text{rec}}]  \\
    &\geq \frac{1}{N}\sum_{i=1}^{N}[\frac{1}{2}||g(f(x_{i_1}^{adv}))-x_{i_2}^{adv}+g(f(x_{i_1}))-g(f(x_{i_1})) + g(f_g(x_{i_2})) -x_{i_2}||^2 -c_{\text{rec}} ] \\
    &\geq \frac{1}{2N}\sum_{i=1}^{N}[||g(f(x_{i_1}^{adv}))-x_{i_2}^{adv}+g(f(x_{i_1}))-g(f(x_{i_1})) + g(f_g(x_{i_2})) -x_{i_2}||^2 -2c_{\text{rec}} ] \\
    & \geq \frac{1}{2N}\sum_{i=1}^{N}[||g(f(x_{i_1}))-x_{i_2}||^2+||g(f(x_{i_1}^{adv}))-g(f(x_{i_1}))||^2+||g(f_g(x_{i_2})-x_{i_2}^{adv}||^2-2c_{\text{rec}}]          \\
    &\geq \frac{1}{2}\mathcal{L}_{\text{rec}}+ \frac{1}{2N}\sum_{i=1}^{N}[||g(f(x_{i_1}^{adv}))-g(f(x_{i_1}))||^2+||g(f_g(x_{i_2})-x_{i_2}^{adv}||^2-2c_{\text{rec}}].  
\end{aligned}
\end{equation*}

Since we have  $||g(f_g(x_{i_2})-x_{i_2}^{adv}||^2 = ||g(f_g(x_{i_2})-x_{i_2}+\delta||^2$, and $||\delta||_2$ is tiny, it leads $||g(f_g(x_{i_2})-x_{i_2}+\delta||^2 \geq ||g(f_g(x_{i_2})-x_{i_2})||^2-||\delta||^2$. 
Following the previous study~\citep{cao2022understand}, the decoder $g(\cdot)$ in a MAE can be represented as an interpolation of the encoder output, denoted by $g(f(x)) = \mathbf{A}_{\text{dec}} \mathbf{V}_{\text{enco}}$. $A_{\text{dec}}$ represents the interpolation weights, and $\mathbf{V}_{\text{enco}}$ is the encoder output.

\begin{equation*}
\begin{aligned}
\mathcal{L}_{\text{rec}}^{adv} \geq \frac{1}{2}\mathcal{L}_{\text{rec}}+ \frac{1}{2N}\sum_{i=1}^{N}[||\textbf{A}^{dec}_{adv_i}\mathbf{V}^{enco}_{adv_i} -\textbf{A}^{dec}_i\mathbf{V}^{enco}_i||^2-||\delta||^{2} - c_{\text{rec}}] 
\end{aligned}
\end{equation*}
\begin{equation*}
\begin{aligned}
&= \frac{1}{2}\mathcal{L}_{\text{rec}}+ \frac{1}{2N}\sum_{i=1}^{N}[||\mathbf{A}^{dec}_{adv_i}\mathbf{V}^{enco}_{adv_i} -\mathbf{A}^{dec}_i\mathbf{V}^{enco}_i + \mathbf{A}^{dec}_{i}\mathbf{V}^{enco}_{adv_i}- \mathbf{A}^{dec}_{i}\mathbf{V}^{enco}_{adv_i} ||^2-||\delta||^{2} - c_{\text{rec}}]
\\& \geq  \frac{1}{2}\mathcal{L}_{\text{rec}}+ \frac{1}{2N}\sum_{i=1}^{N}[||\mathbf{V}^{enco}_{adv_i}(\mathbf{A}^{dec}_{adv_i}-\mathbf{A}^{dec}_{i}) ||^2 + ||\mathbf{A}^{dec}_i(\mathbf{V}^{enco}_i - \mathbf{V}^{enco}_{adv_i}) ||^2-||\delta||^{2} - c_{\text{rec}}].
\end{aligned}
\end{equation*}

Since $||\mathbf{A}^{dec}_i(\mathbf{V}^{enco}_i - \mathbf{V}^{enco}_{adv_i}) ||^2 \ge ||\delta ||^2$ and a ratio constant $C_A$, we can get the following equation:
\begin{equation*}
\begin{aligned}
\mathcal{L}^{adv}_{\text{rec}} \geq \frac{1}{2}\mathcal{L}_{\text{rec}}+ \frac{1}{2N}\sum_{i=1}^{N}[C_{A}||(\mathbf{A}^{dec}_{adv_i}-\mathbf{A}^{dec}_{i})||^2 - c_{\text{rec}}]. 
\end{aligned}
\end{equation*}

Based on Assumption 2, there exists a constant \( H \) such that \( \left\| \mathbf{A}^{\text{dec}}_{\text{adv},i} - \mathbf{A}^{\text{dec}}_i \right\|^2 \approx \frac{H}{T} \sum_{t=1}^{T} \left\| \mathbf{A}^{\text{dec}}_{\text{adv},i,t} - \mathbf{A}^{\text{dec}}_{i,t} \right\|^2. \)
$T$ represents the number of layers in the decoder, and $\mathbf{A}^{dec}_{adv_i,t}$ and $\mathbf{A}^{dec}_{i,t}$ are the attention matrices at the $t$-th decoder layer corresponding to $i$-th adversarial example and the clean example, respectively. Therefore, we can conclude the proof:
\begin{equation*}
\begin{aligned}
\mathcal{L}^{adv}_{\text{rec}} \geq \frac{1}{2}\mathcal{L}_{\text{rec}}+\frac{1}{2NT}\sum_{t=1}^{T}\sum_{i=1}^{N}[HC_{A}||(\mathbf{A}^{dec}_{adv_{i,t}}-\mathbf{A}^{dec}_{i,t})||^2-c_{\text{rec}}]. 
\end{aligned}
\end{equation*}
The proof is complete.

\subsection{Proof of Theorem 4.1}
\label{proof_thm3}

\textit{Proof:}

We assume $\mathcal{L}_{\text{rec}}$ function is convex at the area $[x-\delta,x+\delta]$. For fixed mask $U_0$, we can get the following inequality:
\begin{equation*}
\begin{aligned}
    &\mathcal{L}_{\text{rec}}(x_{adv},U_0) \leq \mathcal{L}_{\text{rec}}(x,U_0) +<\nabla_{x}L(x,U_0),x_{adv}-x> \\
&\iff \mathcal{L}_{\text{rec}}(x_{adv},U_0) - L_{\text{rec}}(x,U_0) \leq <\nabla_{x}L(x,U_0),x_{adv}-x>.
\end{aligned}
\end{equation*}
We define $y^1_{adv}=x_{adv} -\nabla_{x}\mathcal{L}(x_{adv},U_0), $\\
$$
\begin{aligned}
&\iff \mathcal{L}_{\text{rec}}(x_{adv},U_0) - \mathcal{L}_{\text{rec}}(x,U_0) \leq <\frac{x_{adv}-y_{adv}^1}{\lambda},x_{adv}-x>\\
&\Rightarrow \mathcal{L}_{\text{rec}}(x_{adv},U_0) - \mathcal{L}_{\text{rec}}(x,U_0) \leq \frac{<x_{adv}-y_{adv}^1,x_{adv}-x>}{\lambda}\\ 
&= \frac{(x_{adv})^2-x_{adv}x-x_{adv}y_{adv}^1+y_{adv}^1x}{\lambda} \\
&=\frac{2(x_{adv})^2-2x_{adv}x-2x_{adv}y_{adv}^1+2y_{adv}^1x}{2\lambda} \\
&\frac{(X_{adv})^2-2x_{adv}x+(x_{adv})^2 -2y_{adv}x^1_{adv}+x^2-x^2+(y^1_{adv})^2-(y^1_{adv})^2}{2\lambda} \\
&=\frac{||x_{adv}-x||^2_2+||x_{adv}-y_{adv}^1||^2_2-||y^1_{adv}-x||^2_2}{2\lambda}.
\end{aligned}
$$
Therefore, we can get the inequality:
\begin{equation*}
\begin{aligned}
\mathcal{L}_{\text{rec}}(x_{adv},U_0) - \mathcal{L}_{\text{rec}}(x,U_0) \leq \frac{||x_{adv}-x||^2_2-||y_{adv}^1-x||^2_2}{2\lambda}+\frac{\lambda}{2}||\nabla_X\mathcal{L}(x_{adv},U_0)||^2_2.
\end{aligned}
\end{equation*}
Since $||y_{adv}^1-x||^2_2\leq ||clip(y_{adv}^1,\eta)-x||^2_2$, and $\eta$ is the clipping threshold, we define $x_{adv}^1 = clip(y_{adv}^1,\eta)$ and get a derivation:
\begin{equation*}
\begin{aligned}
\mathcal{L}_{\text{rec}}(x_{adv},U_0) - \mathcal{L}_{\text{rec}}(x,U_0) \leq \frac{||x_{adv}-x||^2_2-||x_{adv}^1-x||^2_2}{2\lambda}+\frac{\lambda}{2}||\nabla_x\mathcal{L}(x_{adv},U_0)||^2_2.
\end{aligned}
\end{equation*}

As the same theory, the law is written as the following formulation:
\begin{equation*}
\left\{
\begin{aligned}
& s=0\\
&\mathcal{L}_{\text{rec}}(x_{adv},U_0) - \mathcal{L}_{\text{rec}}(x,U_0) \leq \frac{||x_{adv}-x||^2_2-||x_{adv}^1-x||^2_2}{2\lambda}+\frac{\lambda}{2}||\nabla_X\mathcal{L}(x_{adv}^1,U_0)||^2_2 \\
& s=1\\
&\mathcal{L}_{\text{rec}}(x_{adv}^1,U_0) - \mathcal{L}_{\text{rec}}(x,U_0) \leq \frac{||x_{adv}-x||^2_2-||x_{adv}^2-x||^2_2}{2\lambda}+\frac{\lambda}{2}||\nabla_x\mathcal{L}(x_{adv}^2,U_0)||^2_2\\
& s=2\\
&\mathcal{L}_{\text{rec}}(x_{adv}^2,U_0) - \mathcal{L}_{\text{rec}}(x,U_0) \leq \frac{||x_{adv}-x||^2_2-||x_{adv}^3-x||^2_2}{2\lambda}+\frac{\lambda}{2}||\nabla_x\mathcal{L}(x_{adv}^3,U_0)||^2_2
\\
&\hspace*{7.5cm}......
\\
& s=S\\
&\mathcal{L}_{\text{rec}}(x_{adv}^S,U_0) - \mathcal{L}_{\text{rec}}(x,U_0) \leq \frac{||x_{adv}-x||^2_2-||x_{adv}^{S+1}-x||^2_2}{2\lambda}+\frac{\lambda}{2}||\nabla_x\mathcal{L}(x_{adv}^S,U_0)||^2_2
\end{aligned}\right.
\end{equation*}

We sum all the above terms as follows:
\begin{equation*}
\begin{aligned}
    &\sum^S_{s=0}(\mathcal{L}_{\text{rec}}(x_{adv}^s,U_0)-\mathcal{L}_{\text{rec}}(x,U_0)) \leq \frac{||x_{adv}-x||^2_2-||x_{adv}^{S+1}-x||^2_2}{2\lambda} +\frac{\lambda}{2}\sum_{s=0}^{S}||\nabla_x\mathcal{L}(x_{adv}^{s},U_0)||^2_2
    \\
    &\leq \frac{||x_{adv}-x||^2_2}{2\lambda} +\frac{\lambda}{2}\sum_{s=0}^{S}||\nabla_x\mathcal{L}(x_{adv}^{s},U_0)||^2_2
\end{aligned}
\end{equation*}

The average loss upper bound $\frac{1}{S+1}\sum^S_{s=0}(\mathcal{L}_{\text{rec}}(x_{adv}^s,U_0)-\mathcal{L}_{\text{rec}}(x,U_0))$ is given as the following inequality:
\begin{equation*}
\begin{aligned}
    \frac{1}{S+1}\sum^S_{s=0}(\mathcal{L}_{\text{rec}}(x_{adv}^s,U_0)-\mathcal{L}_{\text{rec}}(x,U_0)) \leq \frac{||x_{adv}-x||^2_2}{2\lambda (S+1)} +\frac{\lambda}{2 (S+1)}\sum_{s=0}^{S}||\nabla_x\mathcal{L}(x_{adv}^{s},U_0)||^2_2.
\end{aligned}
\end{equation*}

Due to the weak convexity of $\mathcal{L}_{\text{rec}}(\cdot)$, it is evident that $\frac{1}{S+1}\sum^S_{s=0}\mathcal{L}_{\text{rec}}(x_{\text{adv}}^s,U_0) \geq \mathcal{L}_{\text{rec}}\left(\frac{1}{S+1}\sum^S_{s=0}X_{\text{adv}}^s,U_0\right)$ and we can obtain the formulation as follows:
\begin{equation*}
\begin{aligned}
    \mathcal{L}_{\text{rec}}\left(\frac{1}{S+1}\sum^S_{s=0}x_{\text{adv}}^s,U_0\right)-\mathcal{L}_{\text{rec}}(x,U_0)) \leq \frac{||x_{adv}-x||^2_2}{2\lambda (S+1)} +\frac{\lambda}{2 (S+1)}\sum_{s=0}^{S+1}||\nabla_x\mathcal{L}(x_{adv}^{s},U_0)||^2_2.
\end{aligned}
\end{equation*}
For any mask $U_e$ and $e\in [1,E]$, we can obtain similar results. Thus, we can draw the results:
\begin{equation*}
\begin{aligned}
    &\frac{1}{(1-\rho)E}\sum_{e=1}^{E}[\mathcal{L}_{\text{rec}}\left(\frac{1}{S+1}\sum^S_{s=0}x_{\text{adv}}^s,U_e\right)-\mathcal{L}_{MAE}(x,U_e))] \leq \frac{1}{(1-\rho)E} \sum_{e=1}^E[\frac{||x_{adv}-x||^2_2}{2\lambda (S+1)}\\ 
    &+\frac{\lambda}{2( S+1)}\sum_{s=0}^{S}||\nabla_x\mathcal{L}(x_{adv}^{s},U_e)||^2_2].
\end{aligned}
\end{equation*}
For a fixed masked rate $p$, the proof has been completed. $\mathcal{L}_{\text{rec}}(X,U_e)) \approx \mathcal{L}_{\text{rec}}^*(X)$ for any $e \in [1,E]$, and we can conclude the proof:
\begin{equation*}
\begin{aligned}
    &\frac{1}{(1-\rho)E}\sum_{e=1}^{E}[\mathcal{L}_{\text{rec}}\left(\frac{1}{S+1}\sum^S_{s=0}x_{\text{adv}}^s,U_e\right)-\mathcal{L}_{\text{rec}}^*(x)] \leq \frac{1}{(1-\rho)E} \sum_{e=1}^E[\frac{||x_{adv}-x||^2_2}{2\lambda (S+1)}\\ 
    &+\frac{\lambda}{2( S+1)}\sum_{s=0}^{S+1}||\nabla_x\mathcal{L}(x_{adv}^{s},U_e)||^2_2].
\end{aligned}
\end{equation*}
The proof is complete.
\section{Details of Mask Autoencoder}

\textbf{MAE:} For CIFAR-10, CIFAR-100, and SVHN, we use Base-MAE~\citep{he2022masked}, setting the image size to 32 and the patch size to 4, while leaving the other parameters unchanged. For ImageNet, we directly used Large-MAE.

For MAE training on CIFAR-10, CIFAR-100, and ImageNet, we use A100 GPUs, training for 2000 epochs with the learning rate of $1e^{-3}$, followed by an additional 1000 epochs with the learning rate of $1.5e^{-4}$. Our training batch size is 64. For ImageNet, we directly use the checkpoint provided by the author. 

\noindent \textbf{MaskDiT:} For SVHN, CIFAR100, and CIFAR10, the network first divides the 32x32 CIFAR-10 image into non-overlapping 8x8 patches, with each patch sized 4x4, resulting in a total of 64 patches. Each patch is mapped to a 128-dimensional embedding space through a linear projection layer, and learnable positional encodings are added. During training, 50\% of the patches are randomly masked, and only the unmasked patches are fed into an encoder composed of 6 Transformer blocks, each consisting of multi-head self-attention and a feed-forward network. Subsequently, the encoded unmasked patches are concatenated with learnable mask tokens and passed into a decoder composed of 3 lightweight Transformer blocks. Finally, a linear projection layer maps the decoded patches back to the 4x4x3 patch space, completing image reconstruction and denoising score prediction.
For ImageNet, we use patch size as 16$\times$ 16 and other parameter are same with original work~\citep{zheng2024fast}. Both network training process and parameters are aligned.

\subsection{Details of Robust fine-tuning}

For CIFAR-10, CIFAR-100, and SVHN, AMRM-Pure$_{\text{MAE}}$ is fine-tuned for 100 epochs, while for ImageNet, it is fine-tuned for 3 epochs. In contrast, MaskDiT is fine-tuned for 25 epochs on CIFAR-10, CIFAR-100, and SVHN, and 1 epoch on ImageNet.

\subsection{Details of Experiment}
\label{supp:det}

\subsubsection{Hyperparameters in Denoising Process}
All hyperparameters related to the denoising process are encompassed in Table~\ref{patch} and \ref{itera}, covering the number of denoising iterations, step size, mask rate, and patch size for each dataset. The asterisk (*) indicates that the step size decays by 0.5 after more than half of the iterations have been completed.
\begin{table}[!h]
\setlength{\tabcolsep}{3.5pt} 
\caption{Hyperparameters of Step Size and Patch Size}
\label{patch}
\small
\centering
\begin{tabularx}{\linewidth}{@{}>{\RaggedRight}X cccc cccc @{}} 
\toprule
\multirow{2}{*}{Dataset} & \multicolumn{4}{c}{Step Size} & \multicolumn{4}{c}{Patch Size} \\
\cmidrule(lr){2-5} \cmidrule(lr){6-9}
& \rotatebox{45}{MAE} & \rotatebox{45}{RMAE} & \rotatebox{45}{MaskDiT} & \rotatebox{45}{RMaskDiT} 
& \rotatebox{45}{MAE} & \rotatebox{45}{RMAE} & \rotatebox{45}{MaskDiT} & \rotatebox{45}{RMaskDiT} \\ 
\midrule
CIFAR-10   & 1* & 1* & 1 & 1 & 4 & 4 & 4 & 4 \\
CIFAR-100  & 1 & 1 & 1 & 1 & 4 & 4 & 4 & 4 \\
SVHN       & 1* & 1 & 1 & 1 & 4 & 4 & 4 & 4 \\
ImageNet   & 1* & 1 & 1 & 1 & 16 & 16 & 16 & 16 \\
\bottomrule
\end{tabularx}

\vspace{2mm}
\end{table}

\begin{table}[htbp]
\setlength{\tabcolsep}{3pt} 
\caption{Hyperparameters of Mask Rate and Iteration Numbers}
\label{itera}
\small
\centering
\begin{tabularx}{\linewidth}{@{}>{\raggedright}X *{8}{S[table-format=1.2]} @{}}
\toprule
\multirow{2}{*}{Dataset} & \multicolumn{4}{c}{Mask Rate} & \multicolumn{4}{c}{Iterations} \\
\cmidrule(lr){2-5} \cmidrule(lr){6-9}
& \rotatebox{45}{MAE} & \rotatebox{45}{RMAE} & \rotatebox{45}{MaskDiT} & \rotatebox{45}{RMaskDiT} 
& \rotatebox{45}{MAE} & \rotatebox{45}{RMAE} & \rotatebox{45}{MaskDiT} & \rotatebox{45}{RMaskDiT} \\
\midrule
CIFAR-10   & 0.25 & 0.25 & 0.50 & 0.50 & 150 & 90  & 25 & 15 \\
CIFAR-100  & 0.30 & 0.25 & 0.50 & 0.50 & 100 & 65  & 20 & 20 \\
SVHN       & 0.30 & 0.25 & 0.50 & 0.50 & 150 & 80  & 20 & 20 \\
ImageNet   & 0.20 & 0.25 & 0.30 & 0.30 & 150 & 20  & 25 & 20 \\
\bottomrule
\end{tabularx}

\vspace{2mm}
\end{table}

\end{document}